\documentclass[lettersize,journal]{IEEEtran}
\IEEEoverridecommandlockouts
\usepackage{amsmath,amsfonts}
\usepackage{algorithmic}
\usepackage{array}
\usepackage[caption=false,font=normalsize,labelfont=sf,textfont=sf]{subfig}
\usepackage{textcomp}
\usepackage{stfloats}
\usepackage{url}
\usepackage{verbatim}
\usepackage{graphicx}

\usepackage{cite}
\usepackage[linesnumbered,ruled,vlined]{algorithm2e}
\usepackage{amsmath,amssymb,amsfonts}
\usepackage{graphicx}
\usepackage{textcomp}
\usepackage[table]{xcolor}
\usepackage{booktabs}
\usepackage{multirow}
\usepackage{amsthm}
\usepackage{hyperref}
\usepackage{cleveref}
\usepackage{tabularx}
\usepackage{makecell}

\definecolor{BestCell}{RGB}{245,226,250}
\newcommand{\bestcell}[1]{\cellcolor{BestCell}\textbf{#1}}

\newtheorem{assumption}{Assumption}

\newtheorem{theorem}{Theorem}
\newtheorem{lemma}{Lemma}

\newtheorem{remark}{Remark}

\Crefname{figure}{Fig.}{Figs.}
\Crefname{equation}{Eq.}{Eqs.}

\def\ie{\emph{i.e.}}

\hypersetup{
    colorlinks=true, 
    linkcolor=blue,  
    filecolor=blue,  
    urlcolor=black,   
    citecolor=blue,  
    pdfborder={0 0 0} 
}

\def\BibTeX{{\rm B\kern-.05em{\sc i\kern-.025em b}\kern-.08em
    T\kern-.1667em\lower.7ex\hbox{E}\kern-.125emX}}
\usepackage{balance}
\begin{document}
\title{AirMoE: Statistic-Augmented Over-the-Air MoE for Collaborative Intelligence}
\author{Wei-Bin Kou, Jingreng Lei, Guangxu Zhu$^*$, Yujiu Yang$^*$
\thanks{Wei-Bin Kou and Yujiu Yang are with Tsinghua Shenzhen International Graduate School, Tsinghua University, Shenzhen, China.}
\thanks{Jingreng Lei is with the Department of Electrical and Electronic Engineering, The University of Hong Kong, Hong Kong, China.}
\thanks{Guangxu Zhu is with Shenzhen Research Institute of Big Data, Shenzhen, China.}
\thanks{Corresponding authors: Guangxu Zhu and Yujiu Yang}
}

\markboth{Journal of \LaTeX\ Class Files,~Vol.~18, No.~9, September~2020}%
{How to Use the IEEEtran \LaTeX \ Templates}

\maketitle

\begin{abstract}
In modern edge intelligence, Mixture of Experts (MoE) are increasingly deployed over wireless cloud-edge networks, as a single edge device lacks sufficient resources to host large-scale models locally. In this distributed architecture, a cloud-hosted pretrained Large Model (LM) acts as a shared backbone for latent feature extraction, while heterogeneous experts deployed across distributed, wirelessly-connected clients collaboratively form the task head. However, deploying MoE over wireless links exposes two coupled bottlenecks. On the one hand, routing which clients to activate generally overloads bandwidth-limited uplinks due to required raw feature transmission. On the other hand, aggregating the activated experts' outputs over wireless links is hindered by channel noise and poor scalability. To break these bottlenecks, we propose a statistic-augmented over-the-air MoE (AirMoE) paradigm. Specifically, on the routing side, each client queries its local Feature Retrieval Library (FRL) with a cloud-broadcast compact query, retrieves a prototype-induced statistic, and reports it digitally to the cloud, drastically reducing uplink traffic; the cloud then selects the most relevant clients by aligning these statistics with the LM-extracted features via Jensen--Shannon (JS) divergence. On the aggregating side, selected experts simultaneously transmit their outputs over the multiple-access channel, which physically computes the reweighted sum via waveform superposition, with reweighting coefficients realized through channel-aware power control. The two mechanisms are thus decoupled both algorithmically and physically. We further provide theoretical analyses on convergence and iteration complexity. Taking semantic segmentation task as an example, extensive experiments demonstrate that AirMoE outperforms MoE baselines and single-model competitors. Ablations further confirm the effectiveness of each incorporated component.
\end{abstract}


\section{Introduction}

The relentless growth of Mixture of Experts (MoE) models has pushed their resource demands far beyond the capability of a single edge device. Consequently, modern edge intelligence increasingly deploys MoE in a cloud--edge collaborative setting to split the computation \cite{11519618,ref:edgeai2,ref:moesurvey,ref:scalingvit,xue2024wdmoe,song2025mixture}. In this distributed architecture, a powerful cloud server hosts a heavyweight, pretrained Large Model (LM) that acts as a fixed, shared backbone responsible for latent feature extraction. Meanwhile, lightweight, specialized experts are deployed across distributed, wirelessly connected clients to collaboratively form the task head. Because each client may have acquired expertise on a specific input distribution (e.g., a particular environment, sensor, or operating condition), the population of clients naturally forms a heterogeneous expert pool. As in any MoE system, harnessing this expert pool hinges on two core operations: \emph{routing}, which selects the most relevant experts for each input's LM-extracted features, and \emph{aggregating}, which fuses the selected experts' outputs into the final prediction.

However, once the experts are scattered behind wireless links, these two routine operations turn into two coupled bottlenecks that do not exist in a single-machine MoE. First, routing becomes a communication-dominant problem rather than a purely algorithmic one: naively transporting the high-dimensional latent features required for routing decisions would saturate the bandwidth-limited uplink channel \cite{11551290}. Second, aggregating becomes a multiple-access scalability problem: instead of fusing expert outputs by a weighted sum at one location as in single-machine MoE, the cloud has to collect each selected client's full output over orthogonal uplinks (e.g., TDMA) and fuse them digitally, thus both the uplink cost and the scheduling latency scale linearly with the number of activated experts. The two bottlenecks of the wireless MoE are illustrated in \Cref{fig:bottlenecks}.

\begin{figure}[tp]
\includegraphics[width=\linewidth]{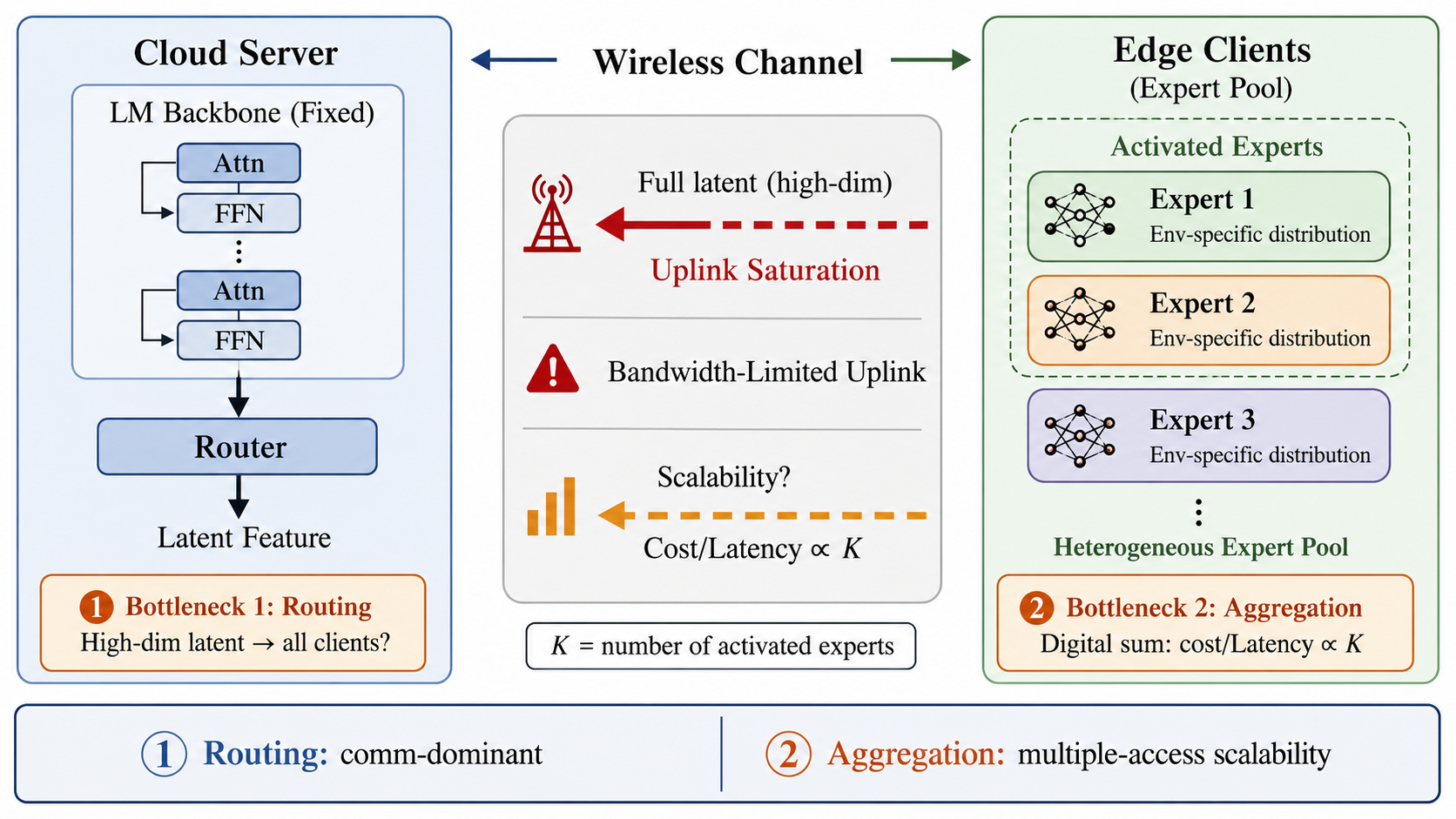}
\vspace{-0.7cm}
\caption{Illustration of the two coupled bottlenecks in wireless MoE.}
\label{fig:bottlenecks}
\vspace{-0.38cm}
\end{figure}

Fortunately, the two bottlenecks exhibit fundamentally different natures, which in turn suggests different remedies.  Routing, on the one hand, is inherently a discrete selection problem. It suffices to select most relevant experts merely based on the compact statistics rather than costly raw features. This works because the routing function is invariant to most of the variation in the raw features. For example, two inputs from the same road condition should be routed identically even if their pixel-level features differ. The statistic deliberately eliminates the variations that are irrelevant to routing decisions in the input. On the other hand, aggregating admits an elegant physical-layer solution via Over-the-air computing (AirComp)~\cite{ref:broadbandagg,ref:airsgd,ref:aircomp}, which leverages the waveform-superposition property of the shared multiple-access channel to sum signals. This property is particularly suited to the wireless MoE aggregating, since, when routed experts transmit their analog-modulated symbols concurrently, the channel can inherently superimpose them into a channel-weighted sum. Instead of transmitting experts' output over orthogonal uplink and fusing them at the cloud, over-the-air aggregation allows all selected experts to transmit output simultaneously and computes their sum by the channel itself. Consequently, both the cost and the latency of the fusion become independent of the number of activated experts.

\begin{figure*}[tp]
\includegraphics[width=\linewidth,height=0.5\linewidth]{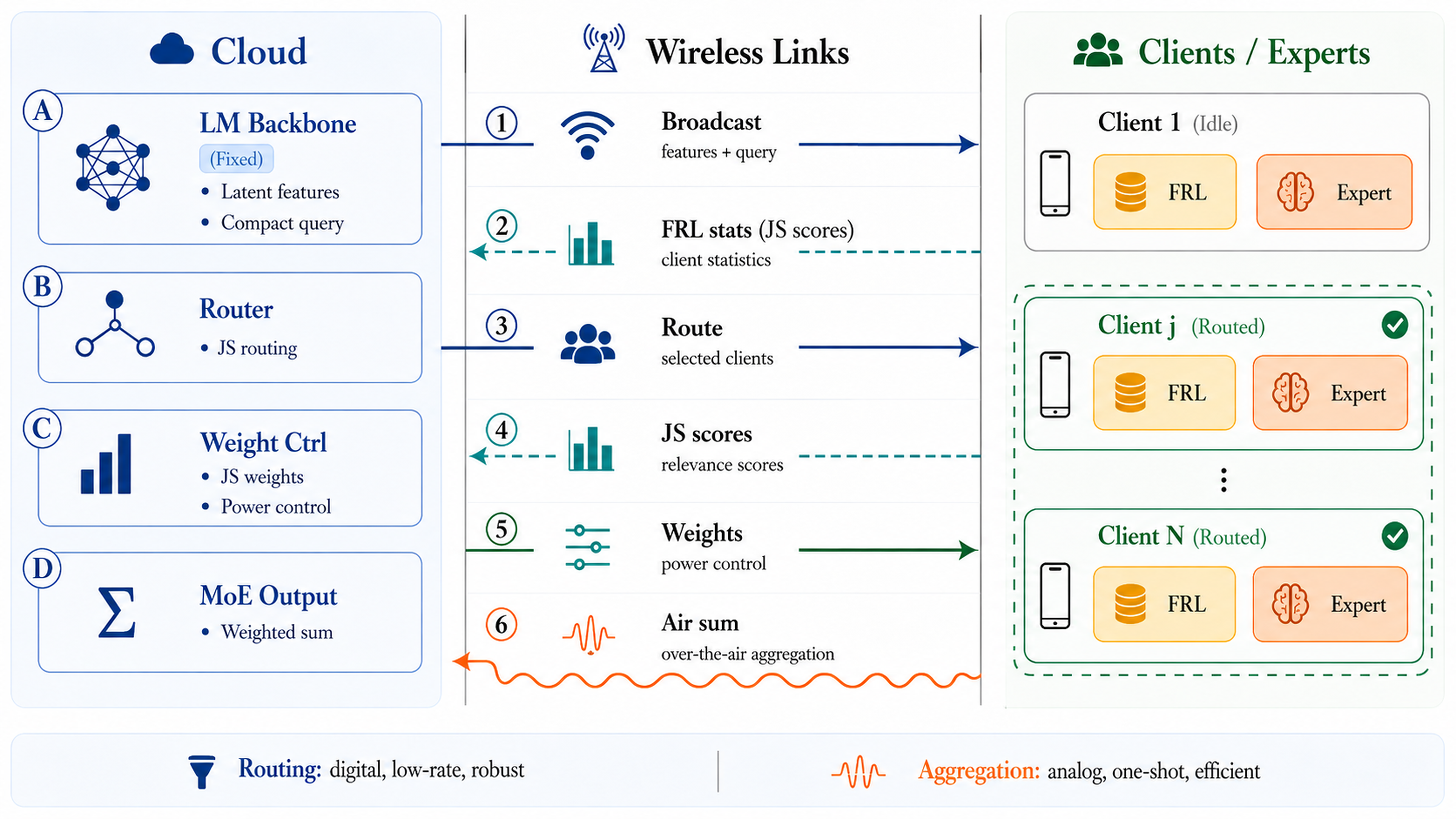}
\centering
\vspace{-0.8cm}
\caption{Overview of the proposed AirMoE.}
\label{fig:AirMoE-overview}
\vspace{-0.3cm}
\end{figure*}

This fundamental dichotomy between wireless MoE's routing and aggregating motivates decoupled design. We address them with individual mechanisms tailored to each characteristic. Specifically, we propose \textbf{AirMoE}, a statistic-augmented, routing-aggregating-decoupled over-the-air MoE paradigm. The framework comprises the following integral components. (i) Cloud-side Feature Extraction and Compact Query Generation: the cloud deploys a pretrained LM backbone to extract latent features for each input and summarizes it into a compact query. Then, the LM-extracted latent features and the summarized query are broadcast to all clients. At client side, the received latent features are used to compute client-wise intermediate features, while the received compact query is used to retrieve clients' expertise. (ii) Client-Wise Feature Retrieval Library (FRL): each client equips the hosted expert with a FRL which stores prototypical feature patterns that encode the client's expertise obtained from its historical experiences. The FRL is maintained by a read-then-update fashion using attention rule based on client-side hidden features. (iii) Statistic-Augmented Routing Mechanism (MoE-RM): as mentioned before, the cloud broadcasts the summarized compact query to all clients and each client locally retrieves a prototype-induced statistic from its equipped FRL based on the cloud query. Subsequently, each client reports the retrieved statistic over a small digital uplink to the cloud. After receiving statistics from all clients, the cloud aligns these clients' statistics with the summarized compact query via Jensen--Shannon (JS) divergence, which is used by the cloud to select the most aligned clients. Finally, the cloud feeds back which clients are routed. (iv) Over-the-Air Aggregating Mechanism (MoE-AM): when receiving the activated feedback from the cloud, each routed client compresses the cloud-broadcast latent features to obtain its own hidden features. Subsequently, it calculates and sends the JS divergence of its hidden features' statistics against the cloud query to the cloud. After receiving all routed clients' JS, the cloud calculates each routed client's aggregation weight based on those JSs, and then sends back the calculated weight to each routed client. Upon receiving their aggregation weights, the routed clients simultaneously transmit power-scaled expert outputs over the multiple-access channel, which physically computes the weighted MoE aggregation via AirComp. In a nutshell, AirMoE decouples routing from aggregating both algorithmically and physically: routing is digital, low-rate, and robust, whereas aggregating is analog, one-shot, and bandwidth efficient. 

In summary, the proposed AirMoE is illustrated in \Cref{fig:AirMoE-overview}, and the main contributions are highlighted as follows: 

\begin{itemize}
\item We formulate MoE over a cloud--edge wireless network, and identify routing and aggregating as the two communication bottlenecks. We propose AirMoE to address these bottlenecks from a statistical perspective by integrating the following integral components: cloud-side feature extraction and compact query generation, client-wise FRL, a statistic-augmented routing mechanism (MoE-RM), and an over-the-air aggregating mechanism (MoE-AM).

\item Client-wise FRL is updated in a read-then-update manner to store each client's expertise as compact prototypes. MoE-RM uses prototype-induced statistics instead of raw features to select clients to reduce uplink saturation, and MoE-AM exploits physical waveform superposition to compute the expert fusion directly over the channel via power control, achieving latency and bandwidth that are invariant to the number of activated clients.

\item We provide theoretical guarantees for AirMoE, including convergence analysis and $\epsilon$-stationarity complexity, in over-the-air setting based on reasonable assumptions.

\item Taking semantic segmentation task as an example, extensive experiments show that AirMoE outperforms MoE and single-model  baselines, and ablations further confirm the effectiveness of the incorporated components.
\end{itemize}

The remainder of this paper proceeds as follows. \Cref{sec:related} reviews related works. \Cref{sec:formulation} formulates the wireless MoE over cloud--edge setting. \Cref{sec:methodology} details the proposed AirMoE paradigm. \Cref{sec:conv} provides theoretical guarantees for the proposed AirMoE. \Cref{sec:experiments} presents experiments and ablations. \Cref{sec:conclusions} concludes this paper.

\section{Related Work} \label{sec:related}

\subsection{Distributed and Edge Foundation Models}
Splitting foundation models over a cloud--edge setting has been studied under split inference and collaborative edge intelligence \cite{11594946,ref:edgeai2,11592738,11535298,11062610}, where an early stage of a deep network runs on the cloud server and a later stage runs on the edge device, or vice versa. Recent works refined this paradigm along several axes. For example, split computing frameworks introduce learnable bottleneck layers to compress the intermediate features transmitted across the split point \cite{ref:bottlenet,ref:splitcompute}; early-exit and adaptive-depth architectures allow inference to terminate at intermediate stages to trade accuracy for latency under fluctuating channel conditions \cite{ref:branchynet}; and device--edge co-inference schemes jointly optimize the partition point, feature compression, and resource allocation to meet end-to-end latency and energy budgets \cite{ref:coinference,ref:jointdnn}. These studies, however, predominantly assume a monolithic network that is merely partitioned across nodes, leaving the deployment of heterogeneous, specialized sub-models on distributed devices largely unexplored. Furthermore, these studies largely overlook the communication constraints. In parallel, recent efforts on edge foundation models pursue model compression, quantization, and parameter-efficient adaptation to fit foundation-scale backbones within constrained hardware \cite{ref:edgefm,ref:tinyllm,Leitrans}. Meanwhile, other parameter-efficient fine-tuning techniques such as low-rank adapters \cite{ref:lora} make it feasible to specialize a frozen backbone with a small number of trainable parameters. 

Our work differs from the above in three ways: (i) we keep the foundation-scale backbone fixed at the cloud and place the adaptation on parallel, distributed MoE experts across devices; (ii) we explicitly account for the wireless links that connect the cloud backbone to the parallel, distributed edge experts, treating communication not as an afterthought but as a core design constraint; and (iii) we decouple the MoE routing and aggregating mechanisms from the underlying wireless channel.

\subsection{MoE and Its Distributed Variant}
MoE has become a pivotal method for scaling models and enhancing task performance \cite{11320560,ref:moesurvey}, and its strength lies in harnessing the specialization of individual experts across diverse data \cite{ref:scalingvit,fan2022m3vit}. The modern sparsely-gated MoE was popularized by \cite{shazeer2017outrageously}, which demonstrated that conditional computation could increase model capacity while keeping per-token compute nearly constant. Subsequent efforts simplified and stabilized the routing mechanism. For instance, Switch Transformers \cite{ref:switch} route each token to a single expert to reduce communication and computation; GShard \cite{ref:gshard} introduced automatic sharding primitives to enable MoE models to scale to billions of parameters across thousands of devices. In addition, GLaM \cite{ref:glam} showed that MoE language models can match or exceed dense counterparts at a fraction of the training and inference cost, and recent open MoE systems such as Mixtral \cite{ref:mixtral} and DeepSeekMoE \cite{ref:deepseekmoe} have further refined expert granularity and load balancing for practical deployment. A primary difficulty in MoE training is balancing expert utilization to avoid load collapse, where a small subset of experts dominates routing \cite{ref:baselayers,ref:expertchoice}. To this end, BASE Layers \cite{ref:baselayers} cast token-to-expert assignment as a balanced linear assignment problem, while Expert Choice routing \cite{ref:expertchoice} inverts the selection so that experts choose tokens, guaranteeing balanced loads. 

Emerging works place experts on different machines or accelerators for capacity scaling \cite{ref:gshard,ref:tutel,li2025theory}. Systems such as Tutel \cite{ref:tutel} provide optimized expert-parallel runtimes, and dedicated communication-scheduling techniques \cite{ref:lina} have been proposed to mitigate the all-to-all bottleneck that dominates distributed MoE execution. However, most of these works assume high-bandwidth, reliable interconnects, and the all-to-all dispatch and aggregation operations are designed under the assumption of essentially lossless, symmetric links. Our considered setting, where experts are hosted at heterogeneous wireless clients and connected to the cloud through noisy and bandwidth-limited links, introduces routing and aggregating challenges that are absent under above the ideal assumptions, and that are the focus of this paper.

\subsection{Over-the-Air Computing (AirComp)}
AirComp leverages the signal-superposition property of the multiple-access channel to compute the weighted sum of distributed data directly during transmission. The information-theoretic foundations of this idea trace back to the study of reliable computation over multiple-access channels \cite{ref:nazergastpar}, where it was shown that uncoded analog transmission can be optimal for estimating functions of correlated sources, in stark contrast to the separation-based design philosophy of digital communication. AirComp has been extensively applied to fast wireless data aggregation and to federated learning (FL) \cite{10907793,10342134,10802073,11246711}, where the global model update is a sum over clients \cite{ref:broadbandagg,ref:airsgd,ref:aircomp}. In the FL setting, over-the-air aggregation has been shown to dramatically reduce communication latency relative to orthogonal multiple access \cite{ref:broadbandagg}, and convergence guarantees have been established even in the presence of channel fading and additive noise \cite{ref:gradientchannel}. Subsequent works have refined this paradigm through joint device-selection and beamforming \cite{ref:airfl}. Power control and signal misalignment have likewise been identified as central design governing the accuracy of the aggregated estimate \cite{ref:powercontrol}.

We are, to our knowledge, the first to map AirComp onto the wireless MoE aggregation, i.e., the expert fusion weights become transmit-power-control coefficients, and the channel performs the MoE fusion. This view also reframes expert sparsity as a communication budget, where the TopK routing not only preserves specialization but also bounds the number of simultaneous transmitters.

\subsection{MoE Routing and Aggregating}
MoE routing is generally achieved by a gating function that assigns inputs to preferred experts. The dominant paradigm remains sparse TopK gating, first introduced in the sparsely-gated MoE layer \cite{shazeer2017outrageously} and later streamlined to single-expert routing in Switch Transformers \cite{ref:switch} to minimize routing overhead. Linear (e.g., softmax) gating \cite{ref:scalingvit,fan2022m3vit} is simple and common, while nonlinear alternatives such as cosine-similarity gating \cite{ref:perturbcosine} project inputs onto a hypersphere and compare to expert embeddings. A recurring concern with sparse gating is training instability and load imbalance, which has motivated auxiliary load-balancing losses \cite{ref:switch}, assignment-based formulations that enforce balanced routing by construction \cite{ref:baselayers,ref:expertchoice}, and stochastic or noise-injected routers that improve exploration over experts \cite{ref:stablemoe}. MoE aggregation typically combines selected experts by weighted sum based on gating-derived weights. Beyond simple weighted sums, recent studies revisit the aggregation mechanism. For instance, hierarchical and multi-level routing structures organize experts into groups before fusion \cite{ref:hashlayers}; analyses of expert specialization show that the quality of the aggregated output depends critically on how routing decisions and combination weights interact \cite{ref:deepseekmoe}. The statistical behavior of these gating and aggregation functions has also been studied for expert estimation \cite{ref:perturbcosine}. 

In this work, we separate routing and aggregating from a statistical perspective, and further decouple their communication modalities, aiming to enhance the MoE overall performance under wireless constraints.

\section{Over-the-Air MoE Formulation}
\label{sec:formulation}

We formulate the wirelessly connected MoE in this section. This formulation makes explicit the two coupled constraints (i.e., uplink routing cost and wireless aggregation scalability) that motivate the proposed AirMoE in \Cref{sec:methodology}. Related key notations are summarized in \Cref{tab:notation}.

\subsection{Network and Computation Setting}
We consider a cloud--edge MoE system consisting of one cloud server and $N$ distributed clients indexed by $j\in\{1,\dots,N\}$, connected over a shared wireless multiple-access channel. The cloud hosts a pretrained LM backbone with parameters $\omega_{\mathrm{LM}}$, while client $j$ hosts an expert $E_j$ with parameters $\theta_j=\theta^{\mathrm{enc}}_j\cup\theta^{\mathrm{dec}}_j$. The collection of clients $\{E_j\}_{j=1}^{N}$ forms a heterogeneous expert pool, where each $E_j$ is specialized on a subset $D_j\subseteq D$ of the data distribution (e.g., a particular weather regime, illumination, or road type), so that $D=\bigcup_{j=1}^{N}D_j$ can reflect the significantly heterogeneous working scenarios of real-world applications. This heterogeneity is precisely what a single centralized model fails to capture, and what the distributed MoE heads are intended to exploit.

\begin{table}[t]
\centering
\setlength{\tabcolsep}{2.5pt}
\caption{Key Notations in the proposed AirMoE}
\label{tab:notation}
\renewcommand{\arraystretch}{1.15}
\begin{tabularx}{\linewidth}{lX}
\toprule
Notation & Explanation \\
\midrule
$D^{(i)}$ & The $i$-th input image of dataset $D$. \\
$\omega_{\mathrm{LM}}$ & Parameters of the cloud LM backbone. \\
$\mathcal{F}^{(i)}_{\mathrm{LM}}$ & LM-extracted latent feature of $D^{(i)}$. \\
$q^{(i)},Q^{(i)}$ & Cloud-side compact query of $\mathcal{F}^{(i)}_{\mathrm{LM}}$ and its normalization. \\
$N$ & Number of clients/experts. \\
$E_j,\theta_j$ & Expert at client $j$ and its parameters. \\
$h^{(i)}_j,\,y^{(i)}_j$ & Intermediate rep.\ and output of expert $j$. \\
$M_j$ & Feature Retrieval Library (FRL) of client $j$. \\
$p_{j,k},w_{j,k}$ & $k$-th prototype and its importance weight in $M_j$. \\
$\alpha^{(i)}_{j,k}$ & Attention weight on $p_{j,k}$ for input $i$. \\
$\tilde{p}^{(i)}_j,P^{(i)}_j$ & Retrieved prototype of client $j$ and its normalization. \\
$\eta$ & FRL memory update rate. \\
$R^{(i)}_j$ & Output-induced distributions. \\
$s^{(i)}_j,\pi^{(i)}_j$ & Routing score and probability of client $j$. \\
$\mathcal{S}^{(i)}$ & Index set of TopK selected clients. \\
$\delta^{(i)}_j$ & JS divergence between $Q^{(i)}$ and $R^{(i)}_j$. \\
$\beta^{(i)}_j$ & Aggregation weight of client $j$. \\
$\gamma_j$ & Uplink channel coefficient of client $j$. \\
$b_j^{(i)}$ & Input-dependent transmit pre-scaling (power control) of client $j$. \\
$\sigma_c^2$ & Complex receiver-noise power at the cloud. \\
$\hat{y}^{(i)}$ & Over-the-air aggregated output for input $i$. \\
\bottomrule
\end{tabularx}
\end{table}

\subsection{Distributed MoE Routing and Aggregating}
For the $i$-th input $D^{(i)}$, the cloud first extracts a shared latent feature $\mathcal{F}^{(i)}_{\mathrm{LM}}$. Based on $\mathcal{F}^{(i)}_{\mathrm{LM}}$, a routing operator $\mathcal{R}$ selects a subset of clients via
\begin{equation}
\!\mathcal{S}^{(i)} \!=\! \mathcal{R}\!\left(\mathcal{F}^{(i)}_{\mathrm{LM}};\{M_j^{(i)}\}_{j=1}^{N}\right)
\subseteq\{1,\dots,N\}, \quad |\mathcal{S}^{(i)}|\!=\!K,
\label{eq:pf_route}
\end{equation}
where $K\ll N$ is the number of routed clients by TopK and $M_j^{(i)}$ denotes any client-side state of input $D^{(i)}$ used by routing. Meanwhile, an aggregation operator $\mathcal{A}$ fuses the selected outputs $\{y_j^{(i)}\}\ (\text{where}\ j \in \mathcal{S}^{(i)})$ into the final prediction $\hat{y}^{(i)}$ through
\begin{equation}
\!\hat{y}^{(i)} \!=\! \mathcal{A}\!\left(\{y^{(i)}_j\}_{j\in\mathcal{S}^{(i)}}\right)
\!=\! \sum_{j\in\mathcal{S}^{(i)}}\!\beta^{(i)}_j y^{(i)}_j,
\quad \sum_{j\in\mathcal{S}^{(i)}}\!\beta^{(i)}_j\!=\!1,
\label{eq:pf_agg}
\end{equation}
where $\{\beta^{(i)}_j\}_{j \in\mathcal{S}^{(i)}}$ are the nonnegative fusion weights. \Cref{eq:pf_agg} exposes the structural property of MoE aggregation, which is a weighted sum and computable over a multiple-access channel.

\subsection{Communication Model}
Realizing \Cref{eq:pf_route}--\Cref{eq:pf_agg} over the wireless network incurs two distinct communication challenges.

\textbf{Prohibitive routing uplink.} To evaluate the routing operator $\mathcal{R}$, for each client $j$ ($j \in \{1, \cdots, N\}$), it must report some descriptor $M^{(i)}_j$ to the cloud over a bandwidth-constrained digital uplink. Let $b(M^{(i)}_j)$ denote its size. The total routing cost per input is
\begin{equation}
C_{\mathrm{route}}^{(i)} = \sum\nolimits_{j=1}^{N} b\!\left(M^{(i)}_j\right).
\label{eq:pf_routecost}
\end{equation}
A naive choice of $M^{(i)}_j$ is the client-side latent features, which makes $C_{\mathrm{route}}^{(i)}$ prohibitive for bandwidth-limited uplink. 

\textbf{Poor scalability of digitally wireless aggregation.} Given $\mathcal{S}^{(i)}$, the selected clients transmit their outputs simultaneously with input-dependent pre-scaling $b_j^{(i)}\in\mathbb{C}$ under the per-client power budget $\mathbb{E}\,\|b_j^{(i)} y^{(i)}_j\|_2^2\le P_0$. Through the uplink channels $\gamma_j\in\mathbb{C}$ and receiver noise $n\sim\mathcal{CN}(0,\sigma_c^2\mathbf{I})$, the cloud-aggregated signal is
\begin{equation}
r^{(i)} = \sum\nolimits_{j\in\mathcal{S}^{(i)}}\gamma_j\,{b_j^{(i)}}\,y^{(i)}_j + n.
\label{eq:pf_rx}
\end{equation}
For the conventional orthogonal digital transmission (e.g., TDMA), the uplink cost and latency scale linearly with the number of activated experts, and becomes non-negligible as the number of the activated clients increases. In addition, realizing \Cref{eq:pf_rx} digitally introduces aggregation distortion. The aggregation distortion can be formulated as
\begin{equation}
\mathcal{E}^{(i)}_{\mathrm{agg}}
= \mathbb{E}\big\|\hat{y}^{(i)}_{\mathrm{ota}} - \hat{y}^{(i)}\big\|_2^2,
\label{eq:pf_distortion}
\end{equation}
where $\hat{y}^{(i)}_{\mathrm{ota}}$ is an estimate of the target aggregation $\hat{y}^{(i)}$. This distortion depends jointly on $\{b_j^{(i)}\}$, $\{\gamma_j\}$, and $\sigma_c^2$. 

\subsection{Optimization Objective of Wireless MoE}
The optimization goal is to jointly design the routing operator $\mathcal{R}$, the fusion weights $\{\beta^{(i)}_j\}$, the input-dependent transmit pre-scalings $\{b_j^{(i)}\}$, the positive receive scalings $\{\rho^{(i)}\}$, and the expert parameters $\{\theta_j\}$ so as to minimize the expected task loss $\mathcal{L}(\cdot)$ subject to the communication budget, physical-layer alignment, and power constraints, \ie,
\begin{align}
\!\!{\min_{\{\theta_j\},\,\mathcal{R},\,\{\beta^{(i)}_j\},\,\{b_j^{(i)}\},\,\{\rho^{(i)}\}}}
& \mathbb{E}_{i}\Big[\mathcal{L}\!\left(\hat{y}^{(i)}_{\mathrm{ota}},Y^{(i)}\right)\Big] \label{eq:pf_obj}\\
\text{s.t.}\quad
& \mathcal{S}^{(i)} = \mathcal{R}\!\left(\mathcal{F}^{(i)}_{\mathrm{LM}};\{M_j\}_{j=1}^N\right),\  \tag{\ref{eq:pf_obj}{a}}\\
& |\mathcal{S}^{(i)}|=K, \tag{\ref{eq:pf_obj}{b}}\\
& C_{\mathrm{route}}^{(i)} \le C_{\max}, \tag{\ref{eq:pf_obj}{c}}\\
& {\mathbb{E}\,\big\|b_j^{(i)}\,y^{(i)}_j\big\|_2^2 \le P_0, ~ j\in\mathcal{S}^{(i)}}, \tag{\ref{eq:pf_obj}{d}}\\
& {\gamma_j b_j^{(i)}=\sqrt{\rho^{(i)}}\,\beta_j^{(i)}, ~ j\in\mathcal{S}^{(i)}}, \tag{\ref{eq:pf_obj}{e}}\\
& \textstyle\sum_{j\in\mathcal{S}^{(i)}}\beta^{(i)}_j = 1,\ \beta^{(i)}_j\ge 0. \tag{\ref{eq:pf_obj}{f}}
\end{align}
Here, $C_{\max}$ is the routing uplink budget and $P_0$ is the per-client power budget. Constraint~(\ref{eq:pf_obj}{e}) explicitly couples the fusion weight and transmit pre-scaling so that, after receiver normalization by $1/\sqrt{\rho^{(i)}}$, the channel coefficient of $y_j^{(i)}$ equals the desired weight $\beta_j^{(i)}$. Since $\hat y_{\mathrm{ota}}^{(i)}$ is the argument of the task loss in \Cref{eq:pf_obj}, the aggregation distortion in \Cref{eq:pf_distortion} affects the objective implicitly rather than being imposed as a separate constraint. Problem~\Cref{eq:pf_obj} is challenging for three reasons: (i) the routing constraint~(\ref{eq:pf_obj}{a}) and ~(\ref{eq:pf_obj}{b}) are combinatorial due to the discrete TopK selection; (ii) the budget~(\ref{eq:pf_obj}{c}) forbids transporting raw features, therefore, $\mathcal{R}$ must operate on highly compressed descriptors; and (iii) the aggregation in the objective is computed through the noisy channel (\Cref{eq:pf_rx}), coupling the learning problem with the physical-layer design via the distortion (\Cref{eq:pf_distortion}).

These three difficulties motivate our decoupled design \textbf{AirMoE} in \Cref{sec:methodology}: a statistic-augmented digital routing mechanism that satisfies (\ref{eq:pf_obj}{a})--(\ref{eq:pf_obj}{c}) by reporting only compact prototype statistics, and a statistic-augmented over-the-air aggregating mechanism that realizes the weighted fusion (\Cref{eq:pf_agg}) directly through the channel (\Cref{eq:pf_rx}), while controlling the distortion (\Cref{eq:pf_distortion}) under the power constraint~(\ref{eq:pf_obj}{d}).

\begin{figure*}[tp]
\includegraphics[width=\linewidth]{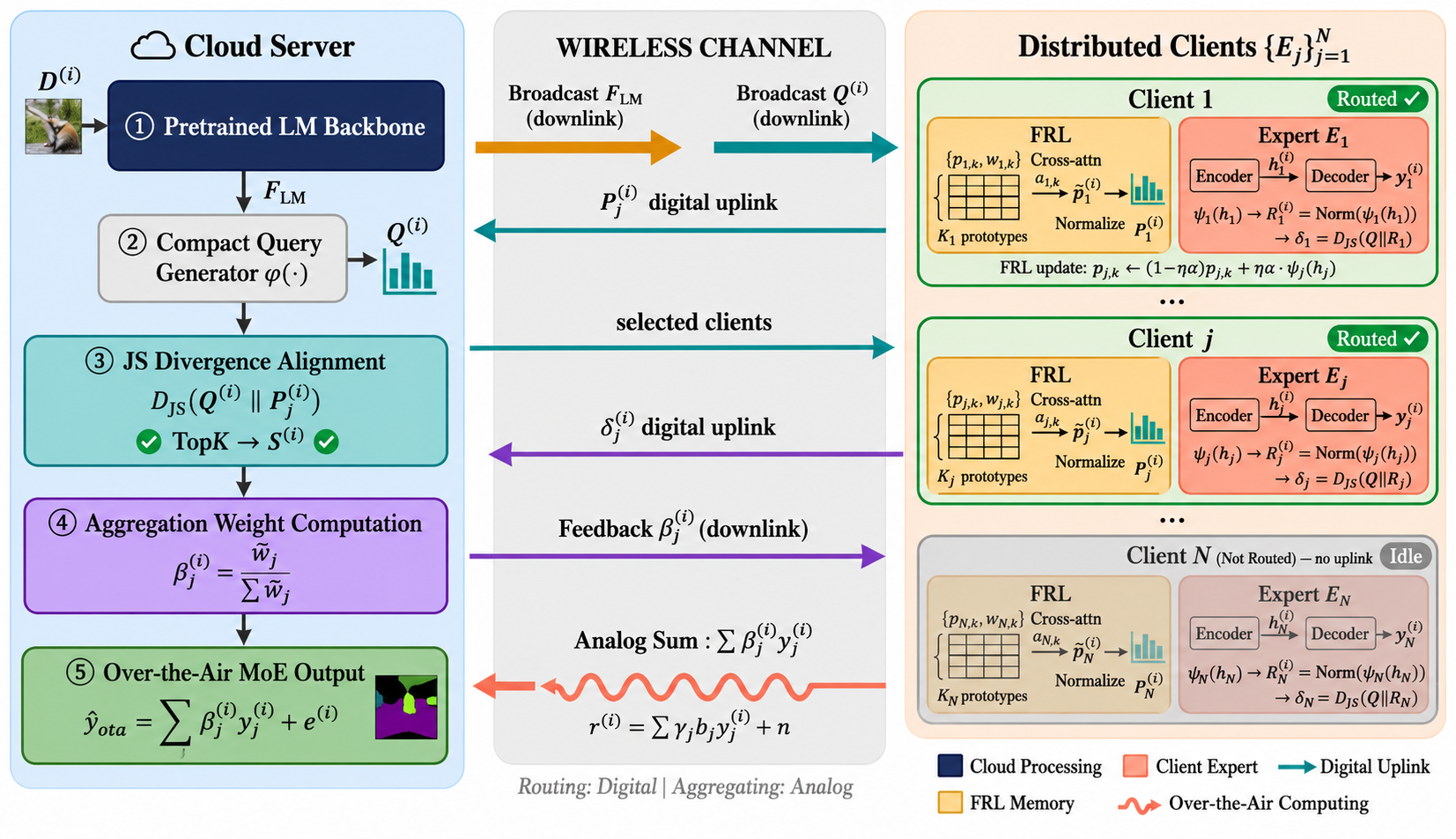}
\centering
\vspace{-1.2cm}
\caption{Detailed illustration of the proposed AirMoE.}
\label{fig:AirMoE-detail}
\vspace{-0.3cm}
\end{figure*}

\section{Methodology} \label{sec:methodology}

The proposed AirMoE consists of (i) cloud-side feature extraction and compact query generation, (ii) client-wise FRL, (iii) statistic-augmented routing mechanism (MoE-RM), and (iv) over-the-air aggregating mechanism (MoE-AM). These components are detailed in \Cref{sec:backbone} to \Cref{sec:moeam}. We then present the training objective of AirMoE in \Cref{sec:obj}. 

\subsection{Cloud-side Feature Extraction and Query Generation}
\label{sec:backbone}
As discussed previously, we consider a cloud server connected to $N$ distributed clients $\{1,\dots,N\}$ over a wireless network, where the cloud hosts a pretrained LM backbone with parameters $\omega_{\mathrm{LM}}$, and client $j$ hosts an expert $E_j$ with parameters $\theta_j$. For the $i$-th input $D^{(i)}$ of dataset $D$, the cloud extracts the latent features $\mathcal{F}^{(i)}_{\mathrm{LM}}$ in a zero-shot manner via
\begin{equation}
\mathcal{F}^{(i)}_{\mathrm{LM}} = \omega_{\mathrm{LM}}\!\left(D^{(i)}\right),
\label{eq:vit}
\end{equation}
where $\mathcal{F}^{(i)}_{\mathrm{LM}}$ is a tensor. To enable both routing and aggregating in a common, communication-friendly space, the cloud summarizes $\mathcal{F}^{(i)}_{\mathrm{LM}}$ into a compact prototype-space query by
\begin{equation}
q^{(i)} = \phi\!\left(\mathcal{F}^{(i)}_{\mathrm{LM}}\right) \in \mathbb{R}^{d},
\qquad
Q^{(i)} = \mathrm{Norm}\!\left(q^{(i)}\right),
\label{eq:query}
\end{equation}
where $\phi(\cdot)$ is a lightweight feature-to-vector map and $\mathrm{Norm}(\cdot)$ converts a vector into a normalized distribution. 

When the extracted features $\mathcal{F}^{(i)}_{\mathrm{LM}}$ and the compact query $Q^{(i)}$ are prepared, the cloud broadcasts them to all clients in a one-to-many manner via the downlink channel. After each client receives them, the extracted features $\mathcal{F}^{(i)}_{\mathrm{LM}}$ is adopted to compute the instant features for input $D^{(i)}$ at each client, and the compact query $Q^{(i)}$ is used to retrieve each client's expertise from client-wise FRL. 

\subsection{Client-Wise Feature Retrieval Library (FRL)}
\label{sec:frl}
Recall that wireless MoE's routing requires the cloud-queried, client-side state information as formalized in \Cref{eq:pf_route}. This client-side state can be its intermediate representation. However, transmitting intermediate representation from clients to the cloud leads to channel saturation and inference latency for bandwidth-constrained uplink owing to their costly size. 

To surmount such uplink challenges, instead of transmitting the raw states, we propose to transmit feature prototype-based expertise from each client to the cloud. To this end,  each client $j$ is proposed to equip its hosted expert $E_j$ with a FRL. The FRL stores prototypical patterns summarizing the client's specialization, which are obtained from its experiencing history in a read-then-update fashion. Specifically, FRL termed as $M_j$ contains $K_j$ entries of client-specific expertise, \ie, $M_j=\{(p_{j,k},w_{j,k})\}_{k=1}^{K_j}$, where $p_{j,k}\in\mathbb{R}^{d}$ are prototype vectors and $w_{j,k}\ge 0$ are importance weights. 

Given the cloud-broadcast query $Q^{(i)}$, client $j$ reads its FRL through cross attention, \ie,
\begin{align}
\alpha^{(i)}_{j,k} &= \frac{\exp\!\big(\mathrm{sim}(Q^{(i)},p_{j,k})\big)}
{\sum_{\ell=1}^{K_j}\exp\!\big(\mathrm{sim}(Q^{(i)},p_{j,\ell})\big)}, \label{eq:attn}\\
\tilde{p}^{(i)}_j &= \sum\nolimits_{k=1}^{K_j}\alpha^{(i)}_{j,k}\,p_{j,k}, \label{eq:read}
\end{align}
where $\mathrm{sim}(\cdot,\cdot)$ is a similarity metric in the shared $d$-dimensional space, thus $Q^{(i)}$ and $p_{j,k}$ are directly comparable and $\tilde{p}^{(i)}_j$ summarizes the most relevant FRL content for input $D^{(i)}$.

When client $j$ is activated and computes a forward pass (\Cref{sec:moeam}), its FRL is renewed in a read-then-update manner through
\begin{align}
p_{j,k} &\leftarrow (1-\eta\alpha^{(i)}_{j,k})\,p_{j,k}
        + \eta\alpha^{(i)}_{j,k}\,\psi_j\!\big(h^{(i)}_j\big), \label{eq:upd_p}\\
w_{j,k} &\leftarrow (1-\eta\alpha^{(i)}_{j,k})\,w_{j,k} + \eta\alpha^{(i)}_{j,k},
\label{eq:upd_w}
\end{align}
where $h^{(i)}_j$ is the client-side intermediate features for input $D^{(i)}$ (as shown in \Cref{eq:pf_hidden}), $\eta\in(0,1]$ is the memory update rate, and $\psi_j:\mathbb{R}^{d_h}\!\to\!\mathbb{R}^{d}$ is a learnable projection into the prototype space. The read-then-update order ensures that the current sample is routed against the pre-updated memory before its information is absorbed, avoiding trivial self-rewriting. Because the FRL resides and operates locally, this update just involves computation but incurs no communication.

\subsection{Statistic-Augmented Digital Routing (MoE-RM)}
\label{sec:moerm}
To select the most relevant clients without uploading raw features through uplink, we route on clients' compact statistics. Specifically, after each client $j$ retrieves its expertise $\tilde{p}^{(i)}_j$ based on the cloud query $Q^{(i)}$, it forms a prototype-induced distribution by normalizing its locally retrieved prototype via
\begin{equation}
P^{(i)}_j = \mathrm{Norm}\!\left(\tilde{p}^{(i)}_j\right),
\label{eq:Pj}
\end{equation}
and uploads it digitally to the cloud via uplink. Since $P^{(i)}_j$ is a low-dimensional normalized vector, this uplink report is much smaller than the size of client-side hidden feature $h^{(i)}_j$. At the cloud, we align $P^{(i)}_j$ with the input-induced cloud-side distribution $Q^{(i)}$ using the Jensen--Shannon (JS) divergence, \ie,
\begin{equation}
\!\!D_{\mathrm{JS}}\!\big(Q^{(i)}\,\|\,P^{(i)}_j\big)
\!=\! \tfrac{1}{2}\Big(D_{\mathrm{KL}}(Q^{(i)}\|I^{(i)}_j)
\!+\! D_{\mathrm{KL}}(P^{(i)}_j\|I^{(i)}_j)\Big),
\label{eq:js_route}
\end{equation}
where $I^{(i)}_j=\tfrac{1}{2}(Q^{(i)}+P^{(i)}_j)$. We prefer JS divergence over raw KL divergence or cosine similarity, because it not only is symmetric and bounded, but also compares two normalized distributions through a shared midpoint, making the routing score less sensitive to scale mismatch across heterogeneous clients. Routing scores and probabilities are computed through
\begin{align}
s^{(i)}_j &= \frac{1}{\epsilon + D_{\mathrm{JS}}(Q^{(i)}\|P^{(i)}_j)}, \label{eq:score}\\
\pi^{(i)}_j &= \frac{\exp(\tau s^{(i)}_j)}{\sum_{\ell=1}^{N}\exp(\tau s^{(i)}_\ell)},
\label{eq:prob}
\end{align}
with stability constant $\epsilon>0$ and temperature $\tau>0$. The cloud applies TopK to $\{\pi^{(i)}_j\}_{j=1}^{N}$, producing the activated client set $\mathcal{S}^{(i)}\subseteq\{1,\dots,N\}$. Only clients in $\mathcal{S}^{(i)}$ proceed to the expert forward pass and following over-the-air MoE aggregation, therefore, TopK controls both expert sparsity and the number of simultaneous uplink transmitters (i.e., the involved communication overheads).

\subsection{Over-the-Air Statistic-Augmented Aggregating (MoE-AM)}
\label{sec:moeam}
In general, for input $D^{(i)}$, only each activated client $j\in\mathcal{S}^{(i)}$ produces its client-side hidden representation $h^{(i)}_j$ and task output $y^{(i)}_j$ via 
\begin{align}
h^{(i)}_j &= E^{\mathrm{enc}}_j\!\big(\mathcal{F}^{(i)}_{\mathrm{LM}};\theta^{\mathrm{enc}}_j\big), \label{eq:pf_hidden} \\
y^{(i)}_j &= E^{\mathrm{dec}}_j\!\big(h^{(i)}_j;\theta^{\mathrm{dec}}_j\big),
\label{eq:pf_expert}
\end{align}
where $\mathcal{F}^{(i)}_{\mathrm{LM}}$ is the cloud LM-extracted latent features that are broadcast by the cloud to all clients, and $E_j=E^{\mathrm{dec}}_j\circ E^{\mathrm{enc}}_j$ represents the expert that resides at client $j$. In contrast, clients that are not in $\mathcal{S}^{(i)}$ remain dummy in processing input $D^{(i)}$.

\textbf{Statistical fusion weights.} To weigh experts by how well their local features agree with the LM-extracted features, each client computes a client-side latent feature-induced distribution $R^{(i)}_j$ and its JS distance to $Q^{(i)}$ through
\begin{align}
R^{(i)}_j &= \mathrm{Norm}\!\big(\psi_j(h^{(i)}_j)\big), \quad j\in\mathcal{S}^{(i)},\label{eq:Rj}\\
\delta^{(i)}_j &= D_{\mathrm{JS}}\!\big(Q^{(i)}\,\|\,R^{(i)}_j\big),\quad j\in\mathcal{S}^{(i)}. \label{eq:delta}
\end{align}
This process generates a small number of statistics $\delta^{(i)}_j$ that each activated client needs to report to the cloud to calculate its aggregation weight. Specifically, at the cloud side, the aggregation weights are reciprocal-distance, normalized over the activated set $\mathcal{S}^{(i)}$, \ie,
\begin{equation}
\tilde{w}^{(i)}_j = \frac{1}{\epsilon+\delta^{(i)}_j},
\qquad
\beta^{(i)}_j = \frac{\tilde{w}^{(i)}_j}{\sum_{\ell\in\mathcal{S}^{(i)}}\tilde{w}^{(i)}_\ell},
\quad j\in\mathcal{S}^{(i)}.
\label{eq:beta}
\end{equation}
After these aggregation weights are already calculated, the cloud feeds back the scalar $\beta^{(i)}_j$ to each activated client with a small downlink overhead. Each activated client $j$ uses its individual weight $\beta^{(i)}_j$ to control its transmission power for over-the-air computing the MoE aggregation. 

Notably, using the same normalized comparison space for $Q^{(i)}$, $P^{(i)}_j$, and $R^{(i)}_j$ keeps routing and aggregating decoupled while remaining comparable at the distribution level.

\textbf{Over-the-air computing the MoE aggregation.} The target fusion of MoE is the weighted sum $\hat{y}^{(i)}$ through \Cref{eq:pf_agg}. Each activated client transmits its analog-modulated output $y^{(i)}_j$ pre-scaled by $b_j^{(i)}$ via channel $\gamma_j$. We realize the MoE summation $r^{(i)}$ physically via waveform superposition as \Cref{eq:pf_rx}, and combine channel-inversion power control (through $b_j^{(i)}$ and $\gamma_j$) with the fusion weight $\beta^{(i)}_j$ via
\begin{equation}
{b_j^{(i)} = \frac{\sqrt{\rho^{(i)}}\,\beta^{(i)}_j}{\gamma_j},
\qquad
\rho^{(i)} = \min_{j\in\mathcal{S}^{(i)}}\frac{P_0\,|\gamma_j|^2}{(\beta^{(i)}_j)^2\,\mathbb{E}\|y^{(i)}_j\|_2^2},}
\label{eq:pc}
\end{equation}
where $P_0$ is the per-client power budget and $\rho^{(i)}$ is the largest receive scaling that satisfies the vector power constraint~(\ref{eq:pf_obj}{d}). The first equality is exactly the physical-layer alignment constraint~(\ref{eq:pf_obj}{e}); hence, it prevents $b_j^{(i)}$ and $\beta_j^{(i)}$ from being optimized independently. Substituting \Cref{eq:pc} into \Cref{eq:pf_rx}, coherently projecting the received signal onto its real component, and applying receiver scaling $1/\sqrt{\rho^{(i)}}$ yield the over-the-air estimate of the MoE aggregation as
\begin{equation}
{\hat{y}^{(i)}_{\mathrm{ota}} = \frac{\operatorname{Re}\{r^{(i)}\}}{\sqrt{\rho^{(i)}}}
= \underbrace{\sum\nolimits_{j\in\mathcal{S}^{(i)}}\beta^{(i)}_j\,y^{(i)}_j}_{\hat{y}^{(i)}}
+ \underbrace{\frac{\operatorname{Re}\{n\}}{\sqrt{\rho^{(i)}}}}_{\text{aggregation noise } e^{(i)}}}.
\label{eq:agg_ota}
\end{equation}
This indicates that the channel computes the statistically reweighted fusion in a single transmission, perturbed by zero-mean real-equivalent noise $e^{(i)}$ with covariance $(\sigma_c^2/(2\rho^{(i)}))\mathbf{I}$. Moreover, maximizing $\rho^{(i)}$ under~(\ref{eq:pf_obj}{d}) minimizes the resulting noise-induced distortion in \Cref{eq:pf_distortion}. Unlike orthogonal digital transmission, the over-the-air scheme in \Cref{eq:agg_ota} uses a single channel regardless of $|\mathcal{S}^{(i)}|$, its latency and bandwidth therefore are invariant to the number of activated experts.

To handle clients in deep fade (where small $|\gamma_j|$ would force $\rho^{(i)}\!\to\!0$ and amplify noise), we adopt truncated channel inversion. Specifically, clients with $|\gamma_j|^2<\gamma_{\mathrm{th}}$ are pruned from $\mathcal{S}^{(i)}$ before $\rho^{(i)}$ is computed. This truncation couples cleanly with MoE-RM, because weak experts are already unlikely to be routed, the truncation therefore rarely removes high-relevance experts.

\subsection{Overall Training Objective of AirMoE}
\label{sec:obj}
The proposed AirMoE is trained in an end-to-end manner with the over-the-air estimated output $\hat{y}^{(i)}_{\mathrm{ota}}$ (i.e., the channel is considered in the forward graph). Taking semantic segmentation as the example task, the total objective $\mathcal{L}$ combines the cross-entropy loss $\mathcal{L}_{\mathrm{CE}}$, the MoE load-balancing term $\mathcal{L}_{\mathrm{LB}}$, and the FRL regularizer $\mathcal{L}_{\mathrm{FRL}}$, \ie,
\begin{equation}
\mathcal{L} = \mathcal{L}_{\mathrm{CE}} + \lambda_{\mathrm{LB}}\mathcal{L}_{\mathrm{LB}}
+ \lambda_{\mathrm{FRL}}\mathcal{L}_{\mathrm{FRL}},
\label{eq:loss}
\end{equation}
where $\lambda_{\text{LB}}, \lambda_{\text{FRL}}\ge 0$ control regularization strengths. Among these terms, $\mathcal{L}_{\text{LB}}$ follows the standard auxiliary load-balancing regularization widely used in MoE training to encourage more uniform expert utilization~\cite{shazeer2017outrageously,ref:switch}, whereas $\mathcal{L}_{\text{FRL}}$ is a task-specific regularizer introduced in this paper to keep the FRL memory compact and stable. 

In the proposed AirMoE, term $\mathcal{L}_{\text{LB}}$ is introduced to encourage balanced expert usage during training, and is defined as 
\begin{equation}
\mathcal{L}_{\mathrm{LB}} = \sum\nolimits_{j=1}^{N} u_j\log u_j - \log N,
\quad u_j=\tfrac{1}{B}\sum\nolimits_{i=1}^{B}\pi^{(i)}_j,
\label{eq:lb}
\end{equation}
where $B$ is the batch size. In this formulation, $\sum\nolimits_{j=1}^{N} u_j \log u_j$ (minimized when $\{u_j\}_{j=1}^N$ is uniform) pushes inputs to be distributed across experts and thereby reduces the risk of expert collapse. Adding $-\log N$ zero-centers the objective at the uniform optimum.

The FRL regularizer keeps client memories compact and stable, \ie,
\begin{equation}
\mathcal{L}_{\mathrm{FRL}} = \sum_{j=1}^{N}\sum_{k=1}^{K_j}\Big(
\underbrace{\|p_{j,k}\|_2^2 + w_{j,k}^2}_{\text{norm/weight decay}}
+ \underbrace{\sum\nolimits_{i=1}^{B}|\alpha^{(i)}_{j,k}|}_{\text{sparse attention}}\Big),
\label{eq:frl}
\end{equation}
where norm/weight decay term prevents unbounded growth of prototype vectors and their importance scalars, controlling scale and avoiding trivial wins by increasing magnitude. Sparse attention term over prototypes encourages retrieving a few relevant prototypes rather than averaging many, which keeps retrieved statistic $P^{(i)}_j$ informative and interpretable.

In summary, AirMoE is illustrated in \Cref{fig:AirMoE-detail} and outlined in \Cref{alg:airmoe}.

\begin{algorithm}[t]
\caption{AirMoE}
\label{alg:airmoe}
\DontPrintSemicolon
\SetKwInOut{Input}{Input}
\SetKwInOut{Output}{Output}
\SetKwInOut{Init}{Init}

\Input{Dataset $D$; cloud LM $\omega_{\mathrm{LM}}$; experts $\{E_j\}_{j=1}^{N}$ with $\{\theta_j\}_{j=1}^{N}$; $K,\tau,\epsilon,\eta,\gamma_{\mathrm{th}},P_0,\lambda_{\mathrm{LB}},\lambda_{\mathrm{FRL}}$.}
\Init{FRLs $M_j=\{(p_{j,k},w_{j,k})\}_{k=1}^{K_j}$, $\forall j$.}
\Output{Trained $\{\theta_j\}$, updated FRLs.}
\BlankLine

\For{\textnormal{each mini-batch} $\{D^{(i)}\}_{i=1}^{B}$}{
  \ForEach{\textnormal{input} $D^{(i)}$}{
    \tcp{\textbf{Cloud feature extraction \& query generation}}
    Extract $\mathcal{F}^{(i)}_{\mathrm{LM}}$, form $Q^{(i)}$ by \Cref{eq:vit,eq:query};
    broadcast $Q^{(i)},\mathcal{F}^{(i)}_{\mathrm{LM}}$\;
    \BlankLine

    \tcp{\textbf{Client-wise FRL read}}
    \ForEach{\textnormal{client} $j\in\{1,\dots,N\}$}{
      Get $\alpha^{(i)}_{j,k},\tilde{p}^{(i)}_j$ by \Cref{eq:attn,eq:read}\;
      Form $P^{(i)}_j$ by \Cref{eq:Pj}; upload digitally\;
    }
    \BlankLine

    \tcp{\textbf{Statistic-augmented routing}}
    Compute $s^{(i)}_j,\pi^{(i)}_j$ by \Cref{eq:js_route,eq:score,eq:prob}\;
    Select $\mathcal{S}^{(i)}\leftarrow\textnormal{Top}K(\{\pi^{(i)}_j\}_{j=1}^{N})$\;
    \BlankLine

    \tcp{\textbf{Expert forward pass, $j\in\mathcal{S}^{(i)}$}}
    Compute $h^{(i)}_j,y^{(i)}_j$ by \Cref{eq:pf_hidden,eq:pf_expert};
    form $R^{(i)}_j,\delta^{(i)}_j$ by \Cref{eq:Rj,eq:delta}; upload $\delta^{(i)}_j$ digitally
    \BlankLine

    \tcp{\textbf{Over-the-air aggregation}}
    Compute $\beta^{(i)}_j$ by \Cref{eq:beta}; feed back to clients\;
    Set $b_j^{(i)},\rho^{(i)}$ by \Cref{eq:pc}\;
    Simultaneous transmission $\Rightarrow r^{(i)}$ by \Cref{eq:pf_rx};
    estimate $\hat{y}^{(i)}_{\mathrm{ota}}$ by \Cref{eq:agg_ota}\;
    \BlankLine

    \tcp{\textbf{FRL renewal (local, $j\in\mathcal{S}^{(i)}$)}}
    Update $p_{j,k},w_{j,k}$ by \Cref{eq:upd_p,eq:upd_w}\;
  }
  \BlankLine

  \tcp{\textbf{End-to-end optimization}}
  Form total loss $\mathcal{L}$ by \Cref{eq:loss};
  update $\{\theta_j\}$ via $\nabla\mathcal{L}$\;
}
\end{algorithm}

\section{Theoretical Guarantees under Channel Noise}
\label{sec:conv}
This section provides theoretical guarantees of AirMoE under channel noise. Because TopK routing makes the system piecewise smooth, we analyze a smooth surrogate following classical nonconvex stochastic approximation. The analysis applies either to a soft-routing relaxation of TopK or to any local region where the activated set $\mathcal{S}^{(i)}$ is fixed. The key new ingredient relative to a single-machine MoE is the over-the-air aggregation noise $e^{(i)}$ in \Cref{eq:agg_ota}.

\subsection{Assumptions}

\begin{assumption}[Smooth surrogate]
\label{as:smooth}
Let $\widetilde{\mathcal{L}}(\Theta,\mathcal{M})=\mathbb{E}_i\big[\ell(\hat y^{(i)}(\Theta),Y^{(i)})\big]$ be the noiseless, relaxed objective, where $\Theta$ collects all trainable parameters, $\ell(\cdot,Y^{(i)})$ is the per-sample loss, and $\mathcal{M}=\{M_j\}_{j=1}^N$ is the FRL state. Assume $\phi(\cdot),\psi_j(\cdot),E^{\mathrm{enc}}_j,E^{\mathrm{dec}}_j,\mathrm{Norm}(\cdot)$ are continuously differentiable, and $\widetilde{\mathcal{L}}(\cdot)$ is lower bounded by $\widetilde{\mathcal{L}}_{\inf}$ and $L$-smooth in $\Theta$, uniformly over bounded $\mathcal{M}$.
\end{assumption}

\begin{assumption}[Bounded states and channel]
\label{as:bounded}
There is $C>0$ with $\|q^{(i)}\|_2\le C$, $\|\psi_j(h^{(i)}_j)\|_2\le C$, $\|p^{(0)}_{j,k}\|_2\le C$, and $0\le w^{(0)}_{j,k}\le 1$. The estimated channel satisfies that the truncation $|\gamma_j|^2\ge\gamma_{\mathrm{th}}$ for all activated clients, so by \Cref{eq:pc} the per-input scaling obeys $\rho^{(i)}\ge\rho_{\min}>0$.
\end{assumption}

\begin{assumption}[Loss and forward-map regularity]
\label{as:loss}
The per-sample loss $\ell(\cdot,Y)$ is twice continuously differentiable in the output, with $L_y$-Lipschitz gradient ($\big\|\nabla_{\hat y}^2\ell(\hat y,Y)\big\|_2\le L_y$) and $\rho_H$-Lipschitz Hessian ($\big\|\nabla_{\hat y}^2\ell(\hat y,Y)-\nabla_{\hat y}^2\ell(\hat y',Y)\big\|_2 \le\rho_H\,\|\hat y-\hat y'\|_2$). The forward-map Jacobian is bounded, $\|J_\Theta(\hat y^{(i)})\|_2\le G$, uniformly over the region of analysis. 
\end{assumption}

\begin{assumption}[Over-the-air noise]
\label{as:noise}
Conditioned on the data and $\Theta$, the aggregation noise $e^{(i)}=\operatorname{Re}\{n\}/\sqrt{\rho^{(i)}}$ in \Cref{eq:agg_ota} is zero-mean, independent of the data, with a symmetric distribution and covariance $\Sigma^{(i)}\preceq(\sigma_c^2/(2\rho_{\min}))\mathbf{I}$, and its dimension (also output dimension) is $d_y$.
\end{assumption}

\begin{assumption}[Stochastic, channel-perturbed update]
\label{as:update}
Define the \emph{noise-smoothed} objective
\begin{equation}
\bar{\mathcal{L}}(\Theta,\mathcal{M})
:=\mathbb{E}_{i}\,\mathbb{E}_{e}\big[\ell(\hat y^{(i)}(\Theta)+e^{(i)},\,Y^{(i)})\big].
\label{eq:smoothed_obj}
\end{equation}
For fixed bounded $\mathcal{M}$, parameters are updated by $\Theta_{t+1}=\Theta_t-\alpha_t H_t g_t$, where $\alpha_t$ is the update step, $g_t$ is the over-the-air stochastic gradient, and $H_t$ is a symmetric positive-definite preconditioning aggregation-scaling matrix applied to the gradient before the parameter step. The preconditioner satisfies $mI\preceq H_t\preceq MI$, and the stepsizes satisfy $\alpha_t>0$, $\sum_t\alpha_t=\infty$, $\sum_t\alpha_t^2<\infty$.
\end{assumption}

\subsection{FRL Stability and Gradient Statistics}

\begin{theorem}[FRL stability]
\label{prop:frl}
Under Assumption~\ref{as:bounded} and $\eta\in(0,1]$, each read-then-update step $t$ keeps every prototype and weight bounded: $\|p^{(t)}_{j,k}\|_2\le C$ and $0\le w^{(t)}_{j,k}\le 1$ for all $t,j,k$.
\end{theorem}

\noindent \emph{Proof.} Refer to Appendix A.

Theorem~\ref{prop:frl} guarantees $\mathcal{M}$ as bounded. By certifying that the read-then-update dynamics keep prototypes inside a ball of radius $C$ and weights inside $[0,1]$ for all rounds $t$, it ensures that the memory cannot diverge or collapse as online updates accumulate. Notably, Theorem~\ref{prop:frl} holds under the mild and practical update rate $\eta\in(0,1]$.

\begin{lemma}[Unbiasedness for the smoothed objective and inflated variance]
\label{lem:unbiased}
Under Assumptions~\ref{as:loss}--\ref{as:update}, the over-the-air gradient is an \emph{unbiased} estimator of the gradient of the smoothed objective,
\begin{equation}
\mathbb{E}[g_t\mid\Theta_t,\mathcal{M}]=\nabla_\Theta\bar{\mathcal{L}}(\Theta_t,\mathcal{M}),
\label{eq:unbiased_smoothed}
\end{equation}
and its variance is bounded by
\begin{equation}
\mathbb{E}\big[\|g_t-\nabla_\Theta\bar{\mathcal{L}}(\Theta_t,\mathcal{M})\|_2^2\big]
{\;\le\;\sigma^2+\kappa\,\frac{\sigma_c^2}{2\rho_{\min}}\;\triangleq\;\sigma_{\mathrm{eff}}^2},
\label{eq:eff_var}
\end{equation}
where $\sigma^2$ is the data-sampling variance and $\kappa:=G^2L_y^2 d_y$.
\end{lemma}

\noindent \emph{Proof.} Refer to Appendix B.

Lemma~\ref{lem:unbiased} is the bridge between the physical-layer channel model and the optimization analysis. It converts the effect of over-the-air aggregation into the two statistical quantities, \ie, an unbiased gradient and a bounded variance $\sigma_{\mathrm{eff}}^2$. It conceptualizes that channel noise, once viewed through the smoothed objective $\bar{\mathcal L}$, contributes no bias to the descent direction and instead manifests purely as inflated variance. Specifically, the standard data-sampling noise variance $\sigma^2$ is additively augmented by the channel term $\kappa\,\sigma_c^2/(2\rho_{\min})$, which cleanly exposes how the complex channel-noise power $\sigma_c^2$, the truncation floor $\rho_{\min}$, and the problem geometry $\kappa=G^2L_y^2 d_y$ each scale the effective noise after coherent real projection. This decomposition lets the convergence theorem invoke off-the-shelf nonconvex-SGD with $\sigma_{\mathrm{eff}}^2$ in place of the usual variance $\sigma^2$, and drives the $\mathcal{O}(\sigma_{\mathrm{eff}}^2/\epsilon^2)$ iteration complexity.

\begin{lemma}[Bias gap to the noiseless objective]
\label{lem:bias}
Under Assumptions~\ref{as:loss}--\ref{as:noise}, for every $\Theta$,
\begin{equation}
\big\|\nabla_\Theta\bar{\mathcal{L}}(\Theta,\mathcal{M})
-\nabla_\Theta\widetilde{\mathcal{L}}(\Theta,\mathcal{M})\big\|_2
{\;\le\;\frac{G\,\rho_H\,d_y}{2}\cdot\frac{\sigma_c^2}{2\rho_{\min}}}
\;\triangleq\;B_\sigma .
\label{eq:bias}
\end{equation}
\end{lemma}

\noindent \emph{Proof.} Refer to Appendix C.

Lemma~\ref{lem:bias} quantifies the price paid for optimizing the tractable smoothed surrogate $\bar{\mathcal L}$ instead of the true objective $\widetilde{\mathcal L}$. It certifies that the two gradients never differ by more than the constant $B_\sigma$, uniformly over all $\Theta$. This complements Lemma~\ref{lem:unbiased}. Specifically, Lemma~\ref{lem:unbiased} guarantees that $g_t$ tracks $\nabla\bar{\mathcal L}$ exactly, and this lemma bounds how far $\nabla_\Theta\bar{\mathcal L}$ itself sits from $\nabla_\Theta\widetilde{\mathcal{L}}$. Therefore, they together ensure that the SGD converges to a stationary point of $\widetilde{\mathcal L}$ bounded by $B_\sigma$, where the bound scales with the constant $\rho_H$, $G$, $d_y$, and $\sigma_c^2/\rho_{\min}$. 

\subsection{Main Result}

\begin{theorem}[Convergence to a stationary point of the smoothed objective]
\label{thm:conv}
Fix any bounded $\mathcal{M}$ satisfying Theorem~\ref{prop:frl}, and suppose Assumptions~\ref{as:smooth}--\ref{as:update} hold with $\bar{\mathcal L}(\cdot)$ $L$-smooth and lower bounded by $\bar{\mathcal L}_{\inf}$ (both inherited from Assumptions~\ref{as:smooth}, \ref{as:loss}). If $\alpha_t\le m/(LM^2)$, then
\begin{align}
\frac{\sum_{t=1}^{T}\alpha_t\,\mathbb{E}\big[\|\nabla_\Theta\bar{\mathcal{L}}(\Theta_t,\mathcal{M})\|_2^2\big]}
     {\sum_{t=1}^{T}\alpha_t}
\;\le\;&\frac{2(\bar{\mathcal{L}}_1-\bar{\mathcal{L}}_{\inf})}{m\sum_{t=1}^{T}\alpha_t}
      + \nonumber \\
      &\frac{LM^2\sigma_{\mathrm{eff}}^2\sum_{t=1}^{T}\alpha_t^2}{m\sum_{t=1}^{T}\alpha_t},
\label{eq:thm}
\end{align}
where $\bar{\mathcal{L}}_1=\bar{\mathcal{L}}(\Theta_1,\mathcal{M})$. Consequently $\liminf_{T\to\infty}$ $\mathbb{E}[\|\nabla_\Theta\bar{\mathcal{L}}(\Theta_T,\mathcal{M})\|_2^2]=0$.
\end{theorem}

\noindent \emph{Proof.} Refer to Appendix D.

Theorem~\ref{thm:conv} certifies that SGD update based on Assumption~\ref{as:update} actually converges to a stationary point of $\bar{\mathcal L}(\cdot)$ despite over-the-air noise. The bound cleanly separates two error sources, where an optimization term $2(\bar{\mathcal L}_1-\bar{\mathcal L}_{\inf})/(m\sum\alpha_t)$ that decays as the accumulated stepsize grows, and a noise term proportional to $\sigma_{\mathrm{eff}}^2\sum\alpha_t^2$ that conditions ($\sum\alpha_t=\infty$, $\sum\alpha_t^2<\infty$) force to vanish. Working together with Lemma~\ref{lem:bias}, Theorem~\ref{thm:conv} upgrades to approximate stationarity of the true objective $\widetilde{\mathcal L}$ whose radius is governed by the channel bias $B_\sigma$.

\begin{remark}[On the fixed-$\mathcal{M}$ assumption]
Theorem~\ref{thm:conv} fixes the FRL state $\mathcal{M}$, whereas in \Cref{alg:airmoe} the prototypes are updated online via \Cref{eq:upd_p}--\Cref{eq:upd_w}. Theorem~\ref{prop:frl} guarantees $\mathcal M$ stays in a bounded set for all $t$. Given a memory rate $\eta$ small relative to $\alpha_t$, the FRL evolves on a slow timescale, and  we can treat $\mathcal{M}$ as quasi-static for the fast $\Theta$-updates. 
\end{remark}

\begin{lemma}[Rate, SNR dependence, and the noiseless-objective floor]
\label{cor:rate}
With $\alpha_t=\alpha/\sqrt{T}$ and $\alpha\le m/(LM^2)$, Theorem~\ref{thm:conv} gives
\begin{align}
\min_{1\le t\le T}\mathbb{E}\big[\|\nabla_\Theta\bar{\mathcal{L}}(\Theta_t,\mathcal{M})\|_2^2\big]
&={\mathcal{O}\!\left(\frac{1}{\sqrt{T}}\Big(\sigma^2+\frac{\kappa\sigma_c^2}{2\rho_{\min}}\Big)\right)} \nonumber \\
&=\mathcal{O}\!\left(\frac{\sigma_{\mathrm{eff}}^2}{\sqrt{T}}\right).
\label{eq:cor_smooth}
\end{align}
Moreover, using $\|\nabla\widetilde{\mathcal L}\|_2^2\le 2\|\nabla\bar{\mathcal L}\|_2^2+2B_\sigma^2$ with the bias $B_\sigma$ of Lemma~\ref{lem:bias}, the iterates approach a \emph{neighborhood} of a stationary point of the \emph{noiseless} objective $\widetilde{\mathcal{L}}$:
\begin{align}
\min_{1\le t\le T}\mathbb{E}\big[\|\nabla_\Theta\widetilde{\mathcal{L}}(\Theta_t,\mathcal{M})\|_2^2\big]
\;\le\;&\mathcal{O}\!\left(\frac{\sigma_{\mathrm{eff}}^2}{\sqrt{T}}\right)
\;+\; \nonumber \\
&{\underbrace{\frac{G^2\rho_H^2 d_y^2}{8}\cdot\frac{\sigma_c^4}{\rho_{\min}^2}}_{\text{channel-induced floor }=\,2B_\sigma^2}}.
\label{eq:cor_floor}
\end{align}
\end{lemma}

\noindent \emph{Proof.} Refer to Appendix E.

Two effects of the channel are now explicit and qualitatively distinct. (i) The transient term in \Cref{eq:cor_floor} is inflated by the variance $\kappa\sigma_c^2/(2\rho_{\min})$. A better channel (larger $\rho_{\min}$, i.e.\ higher receive SNR) reduces the noise that the optimizer must average out, but this term still decays as $T^{-1/2}$. (ii) The steady-state floor $2B_\sigma^2=\mathcal{O}(\sigma_c^4/\rho_{\min}^2)$ does not vanish with $T$. It is the reason why AirMoE converges to a neighborhood of a noiseless stationary point rather than to an exact one. Both terms shrink monotonically with the receive SNR.

\subsection{$\epsilon$-Stationarity and Iteration Complexity}
\label{sec:eps}

The following theorem reports the number of iterations needed to reach an $\epsilon$-stationary point. 

\begin{theorem}[$\epsilon$-stationarity of the smoothed objective]
\label{cor:eps_smooth}
Fix any bounded $\mathcal{M}$ satisfying Theorem~\ref{prop:frl} and suppose Assumptions~\ref{as:smooth}--\ref{as:update} hold. Run the update with the constant stepsize
\begin{equation}
\alpha_t\equiv\alpha
=\min\!\left\{\frac{m}{LM^2},\;
\sqrt{\frac{2\Delta}{LM^2\sigma_{\mathrm{eff}}^2\,T}}\right\},
\label{eq:eps_step}
\end{equation}
where $\Delta:=\bar{\mathcal{L}}_1-\bar{\mathcal{L}}_{\inf}$ denotes the initial-optimality gap of the smoothed objective. Then
\begin{equation}
\min_{1\le t\le T}\mathbb{E}\big[\|\nabla_\Theta\bar{\mathcal L}(\Theta_t,\mathcal{M})\|_2^2\big]
\;\le\;\frac{2LM^2\Delta}{m^2\,T}
\;+\;\frac{2M}{m}\sqrt{\frac{2L\Delta\,\sigma_{\mathrm{eff}}^2}{T}} .
\label{eq:eps_bound}
\end{equation}
Consequently, the iterate with the smallest expected gradient norm is an $\epsilon$-stationary point of $\bar{\mathcal L}$ after at most 
\begin{equation}
T(\epsilon)
=\mathcal{O}\!\left(\frac{L M^2\,\Delta\,\sigma_{\mathrm{eff}}^2}{m^2\,\epsilon^{2}}\right)
={\mathcal{O}\!\left(\frac{1}{\epsilon^{2}}\Big(\sigma^{2}+\frac{\kappa\,\sigma_c^2}{2\rho_{\min}}\Big)\right)}
\label{eq:eps_complexity}
\end{equation}
iterations, where the second equality isolates the dependence on the channel through $\sigma_{\mathrm{eff}}^2=\sigma^2+\kappa\sigma_c^2/(2\rho_{\min})$. 
\end{theorem}

\noindent \emph{Proof.} Refer to Appendix F.

Theorem~\ref{cor:eps_smooth} converts the asymptotic guarantee of Theorem~\ref{thm:conv} into a complexity result. With the stepsize in \Cref{eq:eps_step}, it certifies the best iteration that reaches an $\epsilon$-stationary point of $\bar{\mathcal{L}}(\cdot)$ within a finite number of rounds. Through the two terms of \Cref{eq:eps_bound}, it also recovers the canonical nonconvex-SGD behavior with an $\mathcal O(1/T)$ transient followed by the dominant $\mathcal O(1/\sqrt{T})$ term. In addition, as \Cref{eq:eps_complexity}, by isolating the channel's contribution as an additive $\kappa\sigma_c^2/(2\rho_{\min})$ term inside $\sigma_{\mathrm{eff}}^2$, it makes the communication--computation tradeoff quantitative and actionable. This indicates that the training cost degrades linearly with channel power $\sigma_c^2$ and improves with the receive-SNR floor $\rho_{\min}$, cleanly separating the data-sampling cost $\sigma^2$ from the communication-induced overhead. 

\begin{lemma}[$\epsilon$-stationarity of the noiseless objective, above the floor]
\label{cor:eps_noiseless}
Let $B_\sigma=\tfrac{1}{4}G\rho_H d_y\,\sigma_c^2/\rho_{\min}$ be the bias of Lemma~\ref{lem:bias}. Under the setting of Lemma~\ref{cor:eps_smooth}, the relation $\|\nabla\widetilde{\mathcal L}\|_2^2\le 2\|\nabla\bar{\mathcal L}\|_2^2+2B_\sigma^2$ implies that an $\epsilon$-stationary point of the \emph{noiseless} objective $\widetilde{\mathcal L}$ is reachable \emph{iff}
\begin{equation}
\epsilon\;>\;\epsilon_{\mathrm{floor}}\;:=\;2B_\sigma^2
={\frac{G^2\rho_H^2 d_y^2}{8}\cdot\frac{\sigma_c^4}{\rho_{\min}^2}},
\label{eq:floor}
\end{equation}
in which case it is attained after at most
\begin{equation}
T(\epsilon)=\mathcal{O}\!\left(\frac{L M^2\,\Delta\,\sigma_{\mathrm{eff}}^2}
{m^2\,(\epsilon-\epsilon_{\mathrm{floor}})^{2}}\right)
\label{eq:eps_noiseless_complexity}
\end{equation}
iterations. For $\epsilon\le\epsilon_{\mathrm{floor}}$, no finite $T$ suffices because the channel-induced bias prevents exact stationarity of $\widetilde{\mathcal L}$.
\end{lemma}

\noindent \emph{Proof.} Refer to Appendix G.

Lemma~\ref{cor:eps_noiseless} completes the analytical loop by transferring the finite-time guarantee from the smoothed surrogate $\bar{\mathcal L}$ back to the true objective $\widetilde{\mathcal L}$, and further exposes a fundamental accuracy floor intrinsic to over-the-air training. Its central message is a sharp dichotomy: (i) any target accuracy strictly above the floor $\epsilon_{\mathrm{floor}}=2B_\sigma^2$ is reachable in finite rounds, which shares the same $\mathcal O(\sigma_{\mathrm{eff}}^2/(\epsilon- \epsilon_{\mathrm{floor}})^2)$ complexity as \Cref{eq:eps_noiseless_complexity} but is merely re-centered by the floor. (ii) in contrast, no finite budget can drive $\widetilde{\mathcal L}$ below $\epsilon_{\mathrm{floor}}=2B_\sigma^2$, because the channel bias $B_\sigma$ irreducibly separates stationary points of the surrogate from those of the true loss. 

\section{Experiments and Analyses} \label{sec:experiments}

In this section, we carry out extensive simulations to verify the proposed AirMoE, taking semantic segmentation task as an example. Specifically, we first introduce the simulation setup in \Cref{sec:exp_setup}. We then compare and analyze experimental results in \Cref{sec:exp_comp_analyses}. Subsequently, we carry out evaluations about $\epsilon$-stationarity iteration complexity under over-the-air setting in \Cref{sec:exp_convergence_exploration}. We then visualize the relationship between LM-extracted features and expert-wise FRL prototypes in \Cref{sec:AirMoE_visualization}. Finally, we conduct ablation studies to investigate various hyperparameters' effect on AirMoE's performance in \Cref{sec:ablation_study}.

\subsection{Datasets, Evaluation Metrics and Implementation} \label{sec:exp_setup}
\subsubsection{Datasets}
The \textbf{Cityscapes} dataset \cite{Cordts2016Cityscapes} includes 2,975 training and 500 validation images annotated with masks for 19 semantic classes. The \textbf{CamVid} dataset \cite{brostow2008segmentation} contains 701 images across 11 semantic classes, with 600 used for training and the other 101 for testing. A subset of the \textbf{Apolloscapes} dataset \cite{wang2019apolloscape}, featuring 854 training and 400 test images, provides pixel-level labels for 23 classes. The \textbf{CARLA\_ADV} dataset, generated via the CARLA simulator (version 0.9.13) \cite{dosovitskiy2017carla}, focuses on adverse weather conditions (e.g., fog and rain) and includes 2,764 training and 1,921 test images, annotated for 23 pixel-level classes. These benchmarks include both real-world and simulation scenarios, covering various domain biases. 

\subsubsection{Evaluation Metrics}
We evaluate the proposed AirMoE on semantic segmentation task using four metrics: \textbf{mIoU}, which quantifies the overlap between prediction and ground truth; \textbf{mPre}, which measures the accuracy of positive prediction; \textbf{mRec}, which evaluates the model's ability to identify relevant instances; and \textbf{mF1}, which balances mPre and mRec.

\subsubsection{Implementation Details}
We take ViT \cite{dosovitskiy2020image} as the cloud-side LM backbone to extract shared features and adopt ASPP architecture \cite{chen2017rethinking} for all experts within AirMoE, where the ViT backbone keeps frozen and client-side experts are trained from scratch. The Adam optimizer is chosen for expert optimization with Betas values of (0.9, 0.999), a weight decay of 1e-4, and a learning rate of 3e-4. The default hyperparameters of the proposed AirMoE include: 10 experts, a TopK value of 5, 16 FRL prototypes per expert, a load-balancing weight ($\lambda_{\text{LB}}$) of 0.01, and a FRL regularization weight ($\lambda_{\text{FRL}}$) of 1e-4. 

\begin{table*}[tp]
\setlength{\tabcolsep}{4.9pt}
\centering
\caption{Performance comparison of the proposed AirMoE against other MoE baselines and single-model approaches for all adopted metrics across multiple semantic segmentation datasets}
\begin{tabularx}{\linewidth}{c|cccc|cccc|cccc|cccc}
\toprule
\multirow{2}{*}{Method} & \multicolumn{4}{c|}{Apolloscapes}                                                                            & \multicolumn{4}{c|}{CamVid}                                                                                  & \multicolumn{4}{c|}{CARLA\_ADV}                                                                              & \multicolumn{4}{c}{Cityscapes}                                                                              \\ \cline{2-17} 
                        & \multicolumn{1}{c}{mIoU}  & \multicolumn{1}{c}{mF1}   & \multicolumn{1}{c}{mPre} & \multicolumn{1}{c|}{mRec} & \multicolumn{1}{c}{mIoU}  & \multicolumn{1}{c}{mF1}   & \multicolumn{1}{c}{mPre} & \multicolumn{1}{c|}{mRec} & \multicolumn{1}{c}{mIoU}  & \multicolumn{1}{c}{mF1}   & \multicolumn{1}{c}{mPre} & \multicolumn{1}{c|}{mRec} & \multicolumn{1}{c}{mIoU}  & \multicolumn{1}{c}{mF1}   & \multicolumn{1}{c}{mPre} & \multicolumn{1}{c}{mRec} \\ \midrule
BiSecNetV2              & \multicolumn{1}{c}{22.92} & \multicolumn{1}{c}{27.12} & \multicolumn{1}{c}{-}    & \multicolumn{1}{c|}{-}    & \multicolumn{1}{c}{47.89} & \multicolumn{1}{c}{53.33} & \multicolumn{1}{c}{-}    & \multicolumn{1}{c|}{-}    & \multicolumn{1}{c}{28.80} & \multicolumn{1}{c}{33.59} & \multicolumn{1}{c}{-}    & \multicolumn{1}{c|}{-}    & \multicolumn{1}{c}{33.63} & \multicolumn{1}{c}{43.32} & \multicolumn{1}{c}{-}    & \multicolumn{1}{c}{-}    \\
SegNet                  & \multicolumn{1}{c}{21.01} & \multicolumn{1}{c}{24.60} & \multicolumn{1}{c}{-}    & \multicolumn{1}{c|}{-}    & \multicolumn{1}{c}{46.60} & \multicolumn{1}{c}{50.18} & \multicolumn{1}{c}{-}    & \multicolumn{1}{c|}{-}    & \multicolumn{1}{c}{31.67} & \multicolumn{1}{c}{37.15} & \multicolumn{1}{c}{-}    & \multicolumn{1}{c|}{-}    & \multicolumn{1}{c}{43.13} & \multicolumn{1}{c}{53.87} & \multicolumn{1}{c}{-}    & \multicolumn{1}{c}{-}    \\
SegFormer               & \multicolumn{1}{c}{-}     & \multicolumn{1}{c}{-}     & \multicolumn{1}{c}{-}    & \multicolumn{1}{c|}{-}    & \multicolumn{1}{c}{39.37} & \multicolumn{1}{c}{46.23} & \multicolumn{1}{c}{-}    & \multicolumn{1}{c|}{-}    & \multicolumn{1}{c}{-}     & \multicolumn{1}{c}{-}     & \multicolumn{1}{c}{-}    & \multicolumn{1}{c|}{-}    & \multicolumn{1}{c}{39.37} & \multicolumn{1}{c}{46.23} & \multicolumn{1}{c}{-}    & \multicolumn{1}{c}{-}    \\
AttaNet                 & 20.89                     & 24.85                     & 26.64                    & 25.67                     & 51.12                     & 58.89                     & 58.83                    & 60.96                     & 28.97                     & 34.46                     & 35.63                    & 34.54                     & 22.96                     & 27.28                     & 26.11                    & 30.87                    \\
HRDA                    & 22.19                     & 27.27                     & 34.39                    & 26.94                     & 64.42                     & 75.65                     & 83.66                    & 71.80                     & 34.98                     & 43.10                     & 52.13                    & 40.74                     & 38.89                     & 49.38                     & 64.30                    & 45.80                    \\
TopFormer               & 20.84                     & 24.71                     & 26.37                    & 25.70                     & 48.50                     & 57.02                     & 59.34                    & 57.52                     & 32.32                     & 38.22                     & 41.61                    & 37.20                     & 22.15                     & 26.70                     & 25.48                    & 30.24                    \\
SeaFormer               & 20.58                     & 24.53                     & 25.54                    & 25.28                     & 47.85                     & 56.62                     & 56.22                    & 58.65                     & 28.68                     & 34.24                     & 37.76                    & 33.20                     & 20.51                     & 24.02                     & 23.36                    & 26.85                    \\ \midrule
ViT+ASPP                & \multicolumn{1}{c}{17.11} & \multicolumn{1}{c}{21.18} & \multicolumn{1}{c}{-}    & \multicolumn{1}{c|}{-}    & \multicolumn{1}{c}{68.12} & \multicolumn{1}{c}{77.01} & \multicolumn{1}{c}{-}    & \multicolumn{1}{c|}{-}    & \multicolumn{1}{c}{31.64} & \multicolumn{1}{c}{37.58} & \multicolumn{1}{c}{-}    & \multicolumn{1}{c|}{-}    & \multicolumn{1}{c}{26.31} & \multicolumn{1}{c}{30.37} & \multicolumn{1}{c}{-}    & \multicolumn{1}{c}{-}    \\ \midrule
LinearMoE               & 21.92                     & 26.76                     & 37.23                    & 26.79                     & 71.07                     & 81.36                     & 84.96                    & 79.74                     & 34.97                     & 42.95                     & 52.48                    & 40.57                     & 39.80                     & 51.53                     & 66.07                    & 47.62                    \\
NonLinearMoE            & 22.32                     & 27.30                     & 37.18                    & 27.03                     & 72.09                     & 82.28                     & \bestcell{85.66}                    & 80.40                     & 35.41                     & 43.50                     & 54.08                    & 41.38                     & 41.74                     & 54.07                     & 65.90                    & 50.15                    \\
SoftMoE                 & 22.22                     & 27.11                     & \bestcell{39.02}                    & 26.95                     & 71.54                     & 81.71                     & 85.20                    & 79.54                     & 36.21                     & 44.37                     & 53.39                    & 42.14                     & 42.08                     & 54.45                     & 66.93                    & 49.79                    \\ \midrule
\textbf{AirMoE (Ours)}          & \bestcell{24.75}                     & \bestcell{28.47}                     & 38.35                    & \bestcell{29.48}                     & \bestcell{72.43}                     & \bestcell{82.89}                     & 85.31                    & \bestcell{80.92}                     & \bestcell{37.97}                     & \bestcell{45.56}                     & \bestcell{55.19}                    & \bestcell{43.39}                     & \bestcell{44.74}                     & \bestcell{55.99}                     & \bestcell{68.30}                    & \bestcell{53.47}                    \\ \bottomrule
\end{tabularx}
\label{tab:perf_comp}
\vspace{-0.3cm}
\end{table*}

\begin{figure*}[tp]
\centering
\subfloat[\footnotesize mIoU]{\includegraphics[width=0.5\linewidth, height=0.3\linewidth]{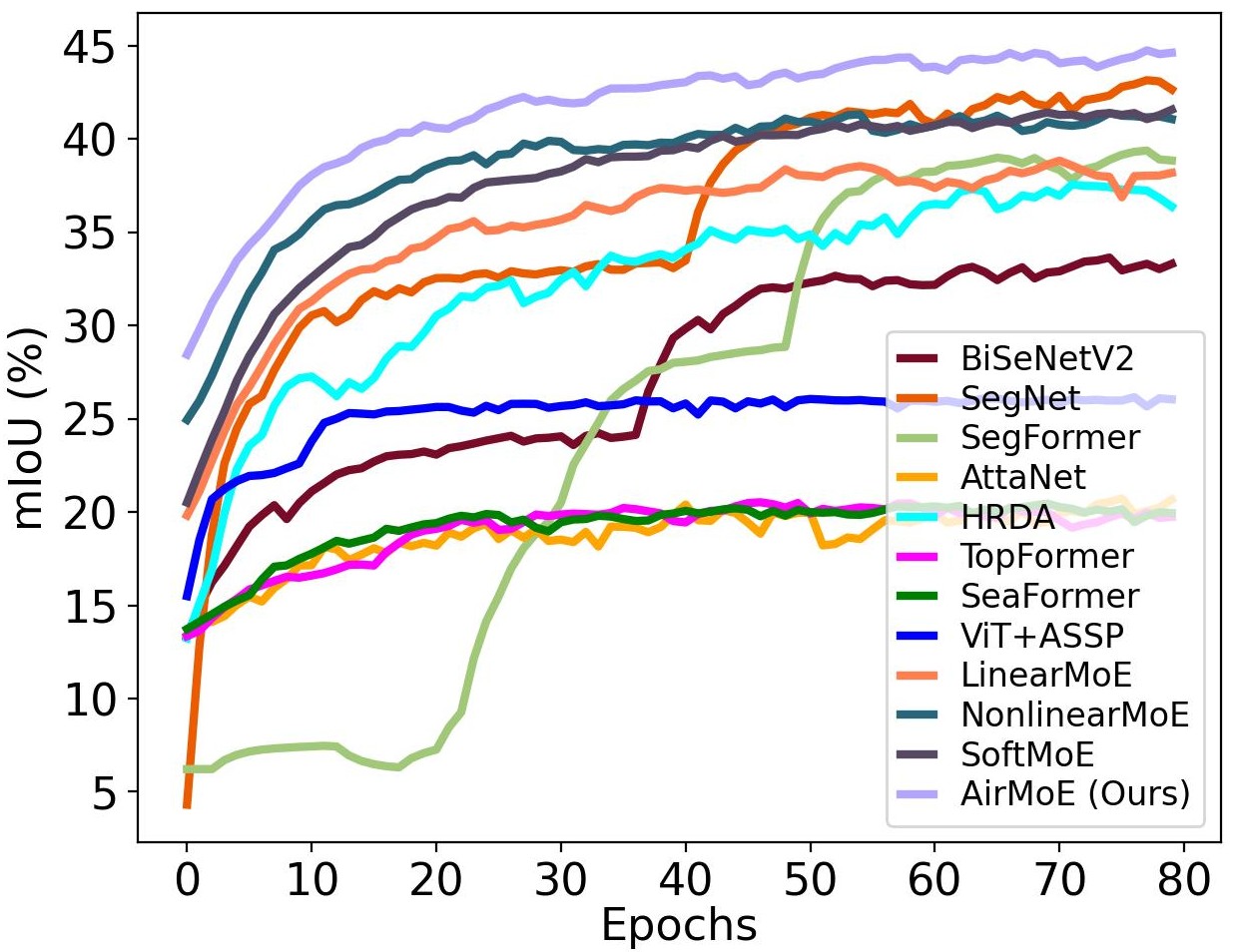}%
\label{Fig:city_mIoU}}
\subfloat[\footnotesize mF1]{\includegraphics[width=0.5\linewidth, height=0.3\linewidth]{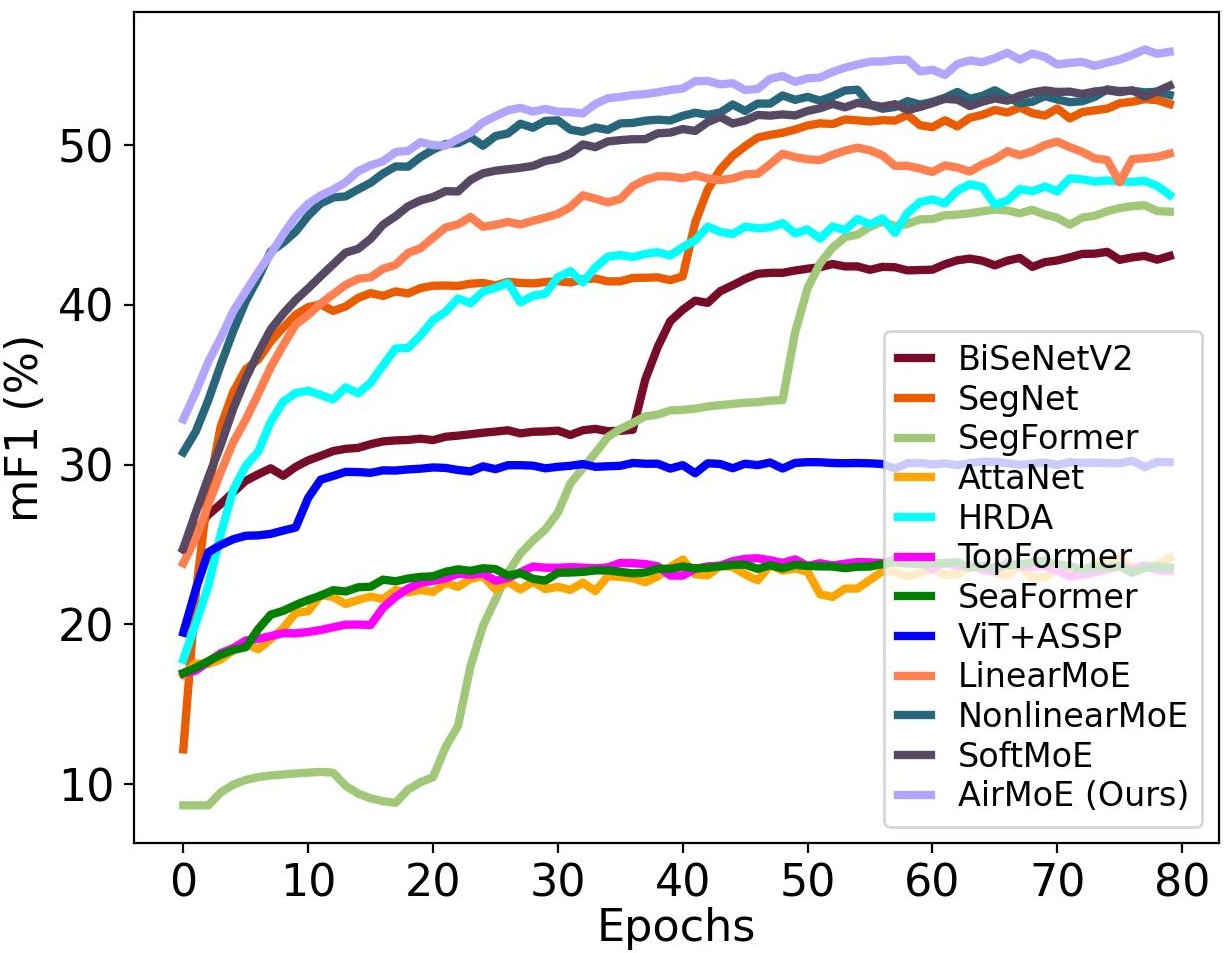}%
\label{Fig:city_mF1}}
\caption{Convergence comparison of the proposed AirMoE against all baselines on Cityscapes dataset.}
\label{Fig:convergece_comp}
\vspace{-0.4cm}
\end{figure*}

For evaluation, on the one hand, we compare the proposed AirMoE with other MoE routing strategies, such as LinearMoE \cite{fan2022m3vit}, NonlinearMoE \cite{nguyen2024statistical}, and SoftMoE \cite{puigcerver2024sparse}. For this MoE-to-MoE comparison, all methods share the same frozen ViT backbone, the same expert architecture, and the same training protocol, yet only the routing/aggregation rule is changed. This isolates the effect of the proposed MoE-RM and MoE-AM. On the other hand, we also compare AirMoE with other single-model methods, such as BiSecNetV2 \cite{yu2021bisenet}, SegNet \cite{badrinarayanan2017segnet}, SegFormer \cite{xie2021segformer}, AttaNet \cite{song2021attanet}, HRDA \cite{hoyer2023domain}, TopFormer \cite{zhang2022topformer}, and SeaFormer \cite{wan2023seaformer}. This comparison is intended to position AirMoE against representative semantic segmentation models.

Notably, the proposed AirMoE, all MoE baselines, and all single-model competitors are implemented using the PyTorch framework and trained using NVIDIA GeForce 4090 GPUs. For MoE-based models, including AirMoE, LinearMoE, NonlinearMoE, and SoftMoE, a pretrained ViT backbone is used and frozen, with only the experts (\ie, ASPPs) within the MoE being trained. In contrast, all single-model competitors are trained from scratch using the aforementioned datasets. Relative to prior MoE baselines, AirMoE introduces extra FRL retrieval and JS-divergence computations, but these operations are performed on compact feature summaries instead of re-running the backbone. 

\subsection{Main Results and Empirical Analyses} \label{sec:exp_comp_analyses}

\subsubsection{Performance Evaluation}
\Cref{tab:perf_comp} compares the main results of the proposed AirMoE against the aforementioned baselines. To interpret the results in \Cref{tab:perf_comp}, instead of reporting the numbers cell-by-cell, we organize the analyses around four research questions (RQs) that progressively isolate the source and nature of AirMoE's advantage across all adopted datasets.

\textbf{RQ1: Does AirMoE deliver consistent gains, or are its improvements confined to particular datasets or metrics?}
AirMoE attains the best mIoU and mF1 on every dataset without exception, improving mIoU over the strongest prior baseline by $+2.43$ on Apolloscapes ($24.75$ vs.\ NonLinearMoE's $22.32$), $+0.34$ on CamVid ($72.43$ vs.\ NonLinearMoE's $72.09$), $+1.76$ on CARLA\_ADV ($37.97$ vs.\ SoftMoE's $36.21$), and $+2.66$ on Cityscapes ($44.74$ vs.\ SoftMoE's $42.08$). Because mIoU and mF1 are the two threshold-independent, class-balanced measures of segmentation quality, their simultaneous first-place ranking across all four benchmarks answers RQ1 affirmatively, \ie, the gains are systematic rather than a dataset-specific or metric-specific artifact.

\textbf{RQ2: Do the improvements stem from the MoE paradigm in general, or to AirMoE's specific design?}
A clear three-tier hierarchy is visible across the method families. The MoE-based models (LinearMoE, NonLinearMoE, SoftMoE, AirMoE) uniformly dominate both the single-model segmentors and the ViT+ASPP baseline, confirming that MoE-based models are inherently better suited to the heterogeneous, multi-domain nature of semantic segmentation scenes. Crucially, within this MoE family, AirMoE still advances the frontier on the primary metrics, which isolates the contribution of its over-the-air aggregation and prototype-guided FRL design from the generic MoE mechanism. RQ2 is thus answered on two levels: MoE certainly helps, and AirMoE's particular design helps beyond that.

\textbf{RQ3: What is the preference of AirMoE's advantage in the precision--recall space?}
The advantage is concentrated in recall. Specifically, AirMoE achieves the highest mRec on all four datasets (e.g.\ $+3.32$ on Cityscapes over SoftMoE and $+2.45$ on Apolloscapes over NonLinearMoE), while its precision is marginally edged out by two cases (\ie, SoftMoE's $39.02$ vs.\ $38.35$ on Apolloscapes, and NonLinearMoE's $85.66$ vs.\ $85.31$ on CamVid). This profile is not a weakness but a favorable trade. The two competing methods raise small precision at the cost of lower recall, whereas AirMoE's balanced expert utilization broadens semantic coverage so that its overall mIoU and mF1 remain superior despite the tiny precision gap. RQ3 therefore reveals that AirMoE's mechanism improves coverage of under-represented regions rather than merely shifting a decision threshold.

\begin{figure*}[tp]
\includegraphics[width=\linewidth, height=0.5\linewidth]{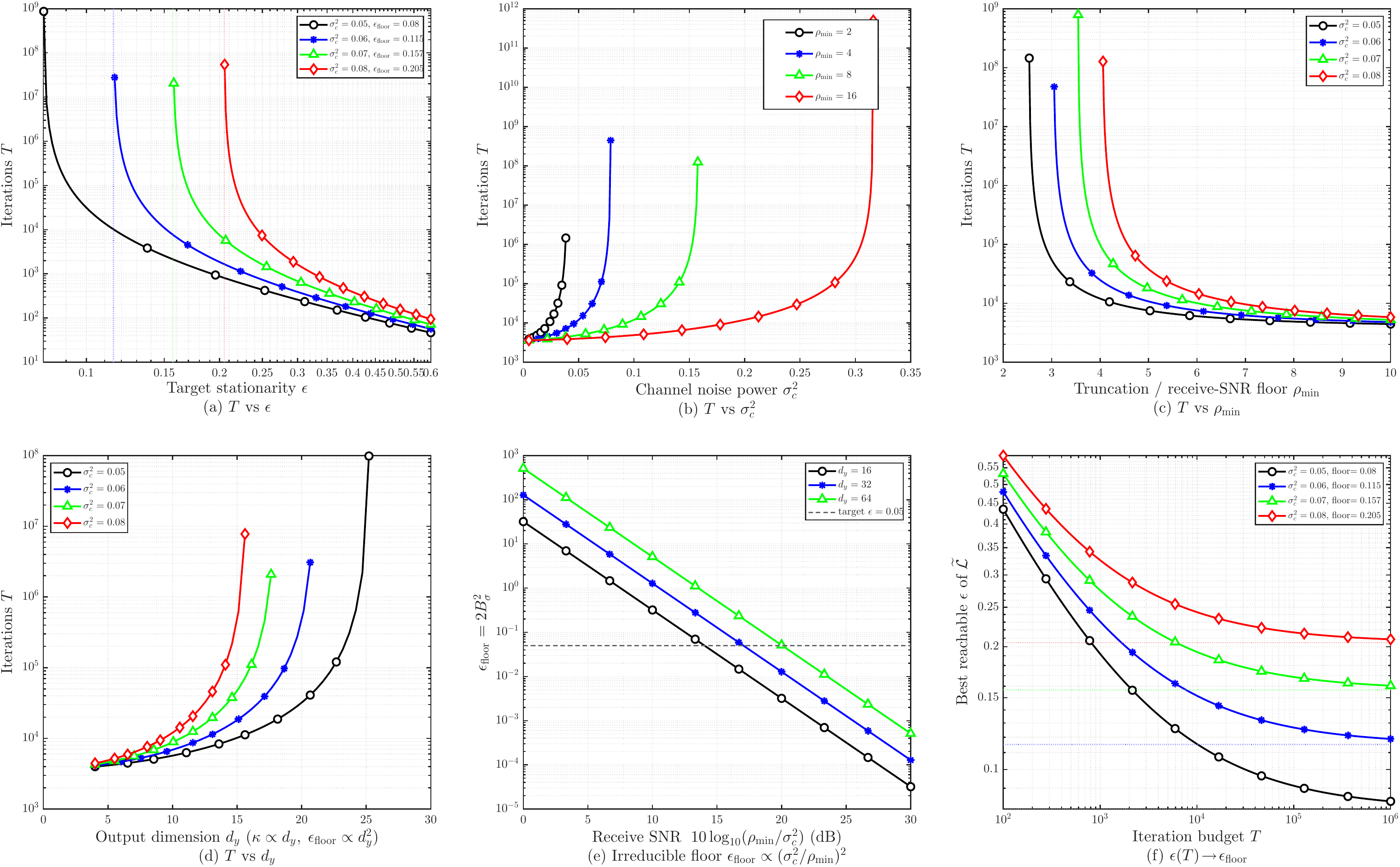}
\vspace{-0.5cm}
\caption{Illustration of the relationship of $\epsilon$-stationarity complexity iteration with respect to various involved factors.}
\label{fig:AirMoE-complexity}
\vspace{-0.3cm}
\end{figure*}

\textbf{RQ4: Where does AirMoE help most, and does this align with practical difficulty?}
Ordering datasets' difficulty by the best baseline mIoU yields Apolloscapes ($22.32$) $>$ CARLA\_ADV ($36.21$) $>$ Cityscapes ($42.08$) $>$ CamVid ($72.09$). AirMoE's relative mIoU improvement tracks this ordering closely: $+10.9\%$ on Apolloscapes ($22.32\!\rightarrow\!24.75$), $+6.3\%$ on Cityscapes ($42.08\!\rightarrow\!44.74$), $+4.9\%$ on CARLA\_ADV ($36.21\!\rightarrow\!37.97$), and near-saturated $+0.5\%$ on the easiest benchmark CamVid ($72.09\!\rightarrow\!72.43$). mRec gains exhibit the same monotone pattern, rising from $+0.52$ on CamVid to $+2.45$ on Apolloscapes and $+3.32$ on Cityscapes over the respective strongest baseline. These cues are the key evidence for RQ4: AirMoE contributes the most precisely on the low-score, high-confusion datasets where competing models saturate or collapse on under-represented classes, and the most on the recall axis that governs coverage of those rare classes. Consequently, AirMoE's benefit scales with task difficulty rather than generating gains on already-easy scenes, which underscores the practical robustness of the proposed AirMoE for real-world semantic segmentation.

\subsubsection{Convergence Comparison}
We also compare the convergence of the proposed AirMoE against other MoE routing baselines and single-model competitors, and the results can be viewed in \Cref{Fig:convergece_comp}. From \Cref{Fig:convergece_comp}, we can derive the following insights: (i) ViT backbone-based methods (including ViT+ASPP, LinearMoE, NonlinearMoE, SoftMoE, and our AirMoE) generally converge faster than single-model approaches. This can be attributed to the pretraining of ViT on a vast number of datasets. (ii) The proposed AirMoE converges faster than other MoE routing baselines, thanks to the statistic-augmented, decoupled routing and aggregating strategies.

\subsection{AirMoE's $\epsilon$-Stationarity Iteration Complexity} \label{sec:exp_convergence_exploration}

This section explores how involved factors impact the $\epsilon$-stationarity iteration complexity of the noiseless objective $\widetilde{\mathcal{L}}$. The specific feature of $\epsilon$-stationarity is that over-the-air aggregation injects a systematic bias whose squared magnitude is $\epsilon_{\mathrm{floor}}$, and the SGD can certify progress only until the expected squared gradient norm reaches this residual level. \Cref{fig:AirMoE-complexity} dissects how each physical and algorithmic factor acts.

\subsubsection{Effect of the target accuracy $\epsilon$}
Panel~(a) of \Cref{fig:AirMoE-complexity} reports the iteration complexity $T$ as a function of the target stationarity accuracy $\epsilon$, with one curve per channel-noise level $\sigma_c^{2}$. The following regimes are visible. When $\epsilon\gg\epsilon_{\mathrm{floor}}$, the term $(\epsilon-\epsilon_{\mathrm{floor}})^{2}$ in the denominator of \Cref{eq:eps_noiseless_complexity} is dominated by $\epsilon^{2}$, thus the bound reduces to the classical non-convex rate $T=\mathcal{O}(\epsilon^{-2})$ and the channel merely rescales the constant. As $\epsilon$ decreases toward the floor, the correction $\epsilon_{\mathrm{floor}}$ becomes non-negligible, the denominator shrinks faster than $\epsilon^{2}$, and $T$ grows super-quadratically. Finally, as $\epsilon$ approaches downward to $\epsilon_{\mathrm{floor}}$, the denominator vanishes and $T\to\infty$, producing the vertical asymptote at $\epsilon_{\mathrm{floor}}$. Furthermore, when $\epsilon\le\epsilon_{\mathrm{floor}}$, no finite $T$ satisfies the bound and the target is infeasible.

This behavior is the algorithmic counterpart of the dichotomy in Lemma~\ref{cor:eps_noiseless}. Over-the-air aggregation contributes a systematic bias to the gradient of the true objective whose squared norm is exactly $\epsilon_{\mathrm{floor}}=2B_\sigma^{2}$. The SGD can only certify progress while the expected squared gradient norm exceeds this bias. Once it falls to the level of the bias, the true gradient is no longer distinguishable from the channel-induced perturbation, and further decrease cannot be guaranteed. The floor is therefore an irreducible accuracy barrier rather than a slow-convergence artifact, and it cannot be overcome by any additional iterations, only by reducing the bias itself.

\subsubsection{Effect of the channel-noise power $\sigma_c^{2}$}
Panel~(b) of \Cref{fig:AirMoE-complexity} fixes the target accuracy $\epsilon$ and sweeps the channel-noise power $\sigma_c^{2}$. This is the decisive factor in \Cref{eq:eps_noiseless_complexity} because it is the only quantity that enters both the numerator $\sigma_{\mathrm{eff}}^{2}$ and the floor $\epsilon_{\mathrm{floor}}$, and it enters them with different orders with compounding effects. Specifically, for the numerator, $\sigma_c^{2}$ appears linearly through the effective variance $\sigma_{\mathrm{eff}}^{2}=\sigma^{2}+\kappa\sigma_c^{2}/(2\rho_{\min})$, reflecting the real-component receiver noise that inflates the variance of the aggregated stochastic gradient. Acting alone, this term would raise $T$ only linearly, exactly as the intrinsic sampling variance $\sigma^{2}$ does. On the other hand, for the floor, $\sigma_c^{2}$ appears with a strictly higher order $\epsilon_{\mathrm{floor}}=2B_\sigma^{2}\propto(\sigma_c^{2})^{2}$, i.e.\ quadratic in the noise power and quartic in the noise amplitude.

The interaction of these two orders explains the shape of each curve. For small $\sigma_c^{2}$, the floor is negligible relative to $\epsilon$, the denominator $(\epsilon-\epsilon_{\mathrm{floor}})^{2}\approx\epsilon^{2}$ is essentially constant, and $T$ rises only linearly through $\sigma_{\mathrm{eff}}^{2}$. As $\sigma_c^{2}$ increases, the quadratic floor grows far faster than the linear numerator; the denominator begins to collapse, and $T$ turns sharply super-linear. Once $\epsilon_{\mathrm{floor}}(\sigma_c^{2})$ reaches the fixed target $\epsilon$, the denominator vanishes, $T\to\infty$ and the curve terminates. This monotone-then-divergent profile quantifies the compounding penalty of noisy analog aggregation, \ie, channel noise degrades not only the convergence rate (via the numerator) but also the best attainable accuracy (via the floor).

\subsubsection{Effect of the SNR floor $\rho_{\min}$}
Panel~(c) of \Cref{fig:AirMoE-complexity} sweeps the receive-SNR floor $\rho_{\min}$ that lower-bounds the effective SNR ratio of the aggregated update. Whereas the channel-noise power $\sigma_c^{2}$ is an environmental quantity that we cannot control, $\rho_{\min}$ is a design parameter, and it enters $\Cref{eq:eps_noiseless_complexity}$ only through the ratio $\sigma_c^{2}/\rho_{\min}$. Raising $\rho_{\min}$ is thus equivalent to improving the effective SNR and acts favorably in both terms of the bound but with different orders. In the numerator, $\rho_{\min}$ suppresses the effective variance linearly, since $\sigma_{\mathrm{eff}}^{2}=\sigma^{2} +\kappa\,\sigma_c^{2}/(2\rho_{\min})$. It accelerates the $T^{-1/2}$ approach to stationarity but saturates once the channel contribution $\kappa\sigma_c^{2}/(2\rho_{\min})$ falls below the intrinsic sampling variance $\sigma^{2}$. In the floor, as $\epsilon_{\mathrm{floor}}\propto\rho_{\min}^{-2}$, $\rho_{\min}$ acts more strongly, because the aggregation bias $B_\sigma$ scales with $\sigma_c^{2}/\rho_{\min}$ and the floor is its square. Increasing $\rho_{\min}$ therefore lowers the irreducible accuracy barrier quadratically. The two effects combine monotonically, \ie, raising $\rho_{\min}$ both shortens the trajectory (via the numerator) and lowers the floor it approaches (via the denominator), thus $T$ decreases throughout. 

The benefit is most pronounced in the high-noise regime, where the floor $\epsilon_{\mathrm{floor}}\propto(\sigma_c^{2}/\rho_{\min})^{2}$ dominates the denominator. Conversely, in the low-noise regime, the floor is already negligible against $\epsilon$, only the linear numerator term remains active, and the marginal value of increasing $\rho_{\min}$ diminishes. This asymmetry identifies $\rho_{\min}$ as the primary design lever under a hostile channel.

\subsubsection{Effect of the output dimension $d_y$}
Panel~(d) of \Cref{fig:AirMoE-complexity} varies the output dimension $d_y$, which enters $\Cref{eq:eps_noiseless_complexity}$ at two distinct orders. In the numerator, it appears linearly through $\kappa=G^{2}L_y^{2}d_y$ in the effective variance $\sigma_{\mathrm{eff}}^{2}$. In the floor, it appears quadratically ($\epsilon_{\mathrm{floor}}\propto d_y^{2}$), because the aggregation bias $B_\sigma$ itself grows linearly with $d_y$ and the floor is the squared bias $2B_\sigma^{2}$. These two orders produce competition. For small $d_y$, the floor is negligible relative to the target $\epsilon$ and the denominator $(\epsilon-\epsilon_{\mathrm{floor}})^{2}\approx\epsilon^{2}$ is nearly constant, thus $T$ rises only linearly through $\sigma_{\mathrm{eff}}^{2}$. As $d_y$ grows, the quadratic floor overtakes the linear numerator and the denominator collapses, thus, $T$ diverges once $\epsilon_{\mathrm{floor}}(d_y)$ reaches the fixed target $\epsilon$.

This behavior exposes a curse of dimensionality intrinsic to over-the-air aggregation. Specifically, transmitting higher-dimensional updates over the analog channel narrows the feasible-accuracy region quadratically in $d_y$, which motivates dimensionality-reducing measures (e.g.\ sparsification or subspace projection) as means of lowering the floor rather than merely the convergence rate.

\subsubsection{The irreducible floor $\epsilon_{\mathrm{floor}}$ vs.\ receive SNR}
Panel~(e) of \Cref{fig:AirMoE-complexity} isolates the irreducible floor in \Cref{eq:floor} and plots it against the receive SNR. Because the floor depends on the channel solely through $\sigma_c^{2}/\rho_{\min}$, with $\epsilon_{\mathrm{floor}}\propto(\sigma_c^{2}/\rho_{\min})^{2} =10^{-\mathrm{SNR}/5}$, its logarithm is affine in the SNR, \ie,
\begin{equation}
    10\log_{10}\epsilon_{\mathrm{floor}}
    =\mathrm{const}-2\,\mathrm{SNR}.
\end{equation}
Therefore, on the log--log axes, each curve is a straight line of slope $-20$\,dB per decade of SNR. Crucially, this slope is universal, and it is set by the squaring of the bias and is independent of the numerator constants $L$, $\Delta$, $M/m$, and $\sigma_{\mathrm{eff}}^{2}$.

The dashed horizontal line marks a representative target accuracy $\epsilon$. Its intersection with a given floor curve determines the minimum receive SNR at which that target is attainable: to the right of the intersection $\epsilon_{\mathrm{floor}}<\epsilon$ and the target is feasible (finite $T$), whereas to the left $\epsilon_{\mathrm{floor}}\ge\epsilon$ and it is unattainable for any iteration budget. This intersection is precisely the SNR at which the vertical asymptote of panel~(a) coincides with the chosen $\epsilon$, therefore, panels~(a) and~(e) describe the same feasibility boundary from complementary viewpoints, \ie, $T$ diverging as $\epsilon$ approaches downward to $\epsilon_{\mathrm{floor}}$ in~(a), and $\epsilon_{\mathrm{floor}}$ crossing $\epsilon$ as the SNR falls in~(e).

In addition, increasing the output dimension $d_y$ shifts every curve upward by $20\log_{10}d_y$\,dB, since $\epsilon_{\mathrm{floor}}\propto d_y^{2}$ without changing the $-20$\,dB/decade slope. Higher-dimensional updates therefore demand a correspondingly higher receive SNR to reach the same target accuracy, which restates the curse of dimensionality of panel~(d) directly in terms of the SNR budget required for feasibility.

\subsubsection{Best reachable accuracy vs.\ iteration budget $T$}
Panel~(f) of \Cref{fig:AirMoE-complexity} takes the complementary view to the preceding panels, \ie, instead of asking how many iterations a target $\epsilon$ requires, it asks for the tightest accuracy attainable given a fixed budget $T$. Inverting \Cref{eq:eps_noiseless_complexity} for $\epsilon$ gives
\begin{equation}
    \epsilon(T)=\underbrace{\epsilon_{\mathrm{floor}}}_{\text{irreducible}}
    +\underbrace{\sqrt{\frac{C\,L\,M^{2}\,\Delta\,\sigma_{\mathrm{eff}}^{2}}
    {m^{2}\,T}}}_{\text{transient, }\;\mathcal{O}(T^{-1/2})}.
    \label{eq:eps-of-T}
\end{equation}
The decomposition separates the two roles. The transient term decays at the classical stochastic rate $T^{-1/2}$ and carries all of the numerator dependence, including the smoothness $L$, the initial gap $\Delta$, the effective variance $\sigma_{\mathrm{eff}}^{2}$, and the squared preconditioner conditioning $M^{2}/m^{2}$. Each of these enlarges the constant and hence slows the descent, but none changes the value the curve approaches. The irreducible term $\epsilon_{\mathrm{floor}}$ is independent of $T$ and sets that limiting value.

The decisive feature is therefore the asymptote (dotted lines). As $T\to\infty$, the transient term vanishes and $\epsilon(T)$ approaches downward to $\epsilon_{\mathrm{floor}}$ rather than $0$. No iteration budget, no matter how large, can drive the noiseless objective below the channel-induced floor. The curve is eventually flat and further iterations are wasted.

This panel thus consolidates the central message of \Cref{fig:AirMoE-complexity}. Two families of parameters play qualitatively distinct and non-interchangeable roles, where the optimization constants $(L,\Delta,\sigma_{\mathrm{eff}}^{2},M/m)$ govern only the speed at which the floor is approached, whereas the channel parameters $(\sigma_c^{2},\rho_{\min},d_y)$ set the floor itself and hence the fundamental limit of achievable stationarity. Improving the former accelerates convergence toward a fixed barrier, and improving the latter is the only way to lower the barrier. Panels~(a) and~(e) locate this barrier from the accuracy and SNR axes, respectively, and panel~(f) shows its consequence along the iteration axis, giving a single coherent account of feasibility for over-the-air aggregation.

\begin{figure}[tp]
\includegraphics[width=0.95\linewidth]{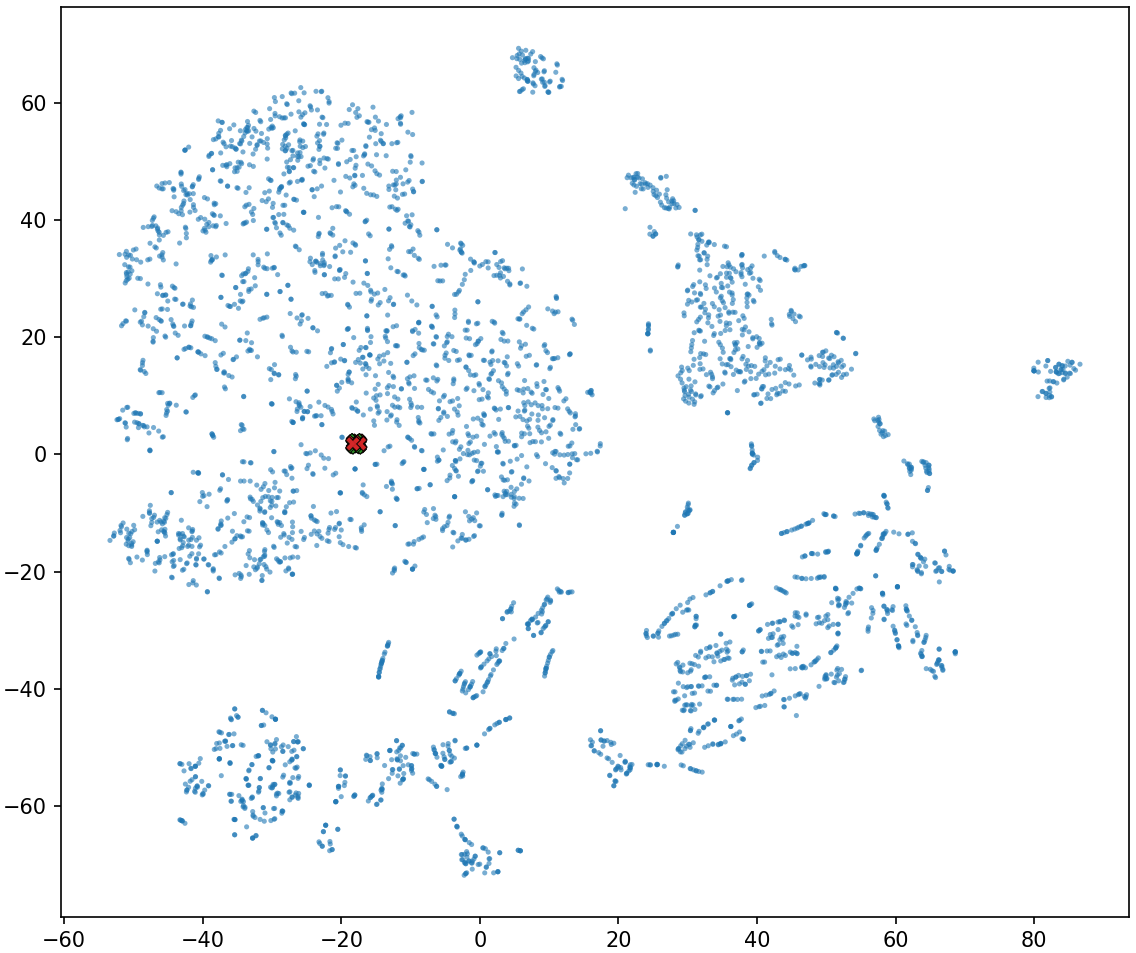}
\vspace{-0.2cm}
\centering
\caption{Relationship visualization between ViT-extracted features and expert-wise FRL prototypes.}
\label{Fig:AirMoE_visualization}
\vspace{-0.3cm}
\end{figure}

\subsection{Visualization of the Relationship between ViT-extracted Features and Expert-wise FRL Prototypes} \label{sec:AirMoE_visualization}
This section visualizes the relationship between cloud-side ViT-extracted features and the most relevant expert's FRL prototypes by t-SNE. The result is presented in \Cref{Fig:AirMoE_visualization}. For clarity, we merely visualize sampled cloud-side shared tokens and the most aligned client's prototypes. Specifically, each plot overlays randomly sampled ViT token features, represented by smaller points, together with the prototypes stored in the single top-routed expert, depicted as larger ``X'' markers. Because each expert FRL contains specific expertise, these ``X'' markers appear quite close in the plot. As indicated by \Cref{Fig:AirMoE_visualization}, the learned prototypes of the most relevant expert align closely with the majority of the ViT-extracted token features, demonstrating the effectiveness of the proposed expert-wise FRLs.

\begin{table}[tp]
\setlength{\tabcolsep}{6.5pt}
\caption{Comparison among cases with different expert numbers}
\begin{tabularx}{\linewidth}{c|c|cccc}
\toprule
\multirow{2}{*}{\begin{tabular}[c]{@{}c@{}}Expert\\ Number\end{tabular}} & \multirow{2}{*}{\begin{tabular}[c]{@{}c@{}}Other Fixed\\ Settings\end{tabular}}                                                                                     & \multicolumn{4}{c}{Cityscapes} \\ \cline{3-6} 
                                                                         &                                                                                                                                                                          & mIoU   & mF1   & mPre  & mRec  \\ \midrule
7                                                                        & \multirow{4}{*}{\begin{tabular}[c]{@{}c@{}}TopK $=5$,\\ FRL prototypes $=16$,\\ $\lambda_{\text{LB}}=0.01$,\\ $\lambda_{\text{FRL}}=10^{-4}$\end{tabular}} & \bestcell{45.45}  & \bestcell{56.76} & \bestcell{71.65} & 53.22 \\
8                                                                        &                                                                                                                                                                          & 45.02  & 56.08 & 70.19 & 53.68 \\
9                                                                        &                                                                                                                                                                          & 44.99  & 55.82 & 69.07 & \bestcell{53.87} \\
10                                                                       &  & 44.74                     & 55.99                     & 68.30                    & 53.47  \\ \bottomrule 
\end{tabularx}
\label{tab:abl_expert_num}
\end{table}

\subsection{Ablation Study} \label{sec:ablation_study}
This section investigates how AirMoE-embedded hyperparameters (including the number of experts, the TopK value, the FRL prototype number, $\lambda_{\text{LB}}$, and $\lambda_{\text{FRL}}$) affect the overall performance of AirMoE. These ablations should be interpreted jointly, where the number of experts and the TopK value together determine routing sparsity, the FRL size controls how sharp or diffuse each expert memory becomes, and $\lambda_{\text{LB}}$ along with $\lambda_{\text{FRL}}$ controls the regularization strength. Notably, we keep the FRL update rate $\eta$ fixed throughout this paper.

\subsubsection{The Influence of the Number of Experts on AirMoE Performance}
\Cref{tab:abl_expert_num} reports the influence of the number of experts on AirMoE. From this table, we can observe the following patterns: (i) As the expert pool grows from $7$ to $10$ under a fixed TopK$=5$, mIoU declines steadily from $45.45$ to $44.74$, mF1 from $56.76$ to $55.99$, and mPre from $71.65$ to $68.30$. This behavior can be interpreted through the activation ratio $K/N$. With $K=5$, enlarging $N$ lowers the activation ratio $K/N$ from $5/7$ to $5/10$, therefore, each expert is selected and updated less frequently, resulting in declined performance. (ii) It is also notable that the degradation is not uniform across metrics, where mRec stays essentially flat and even peaks at $N=9$ ($53.87$), whereas mPre drops the most sharply ($-3.35$). This asymmetry indicates that a larger, more diffuse expert population still covers the relevant semantic regions (preserving recall) but weakens per-region confidence (eroding precision). Because $N=7$ sits at the lower boundary of the search range, we refrain from claiming global optimality. Instead, the consistent decline suggests that AirMoE favors a compact expert pool whose activation ratio remains high.

\begin{table}[tp]
\setlength{\tabcolsep}{7.2pt}
\caption{Comparison among cases with different TopK values}
\begin{tabularx}{\linewidth}{c|c|cccc}
\toprule
\multirow{2}{*}{\begin{tabular}[c]{@{}c@{}}TopK\\ Value\end{tabular}} & \multirow{2}{*}{\begin{tabular}[c]{@{}c@{}}Other Fixed\\ Settings\end{tabular}}                                                                                               & \multicolumn{4}{c}{Cityscapes}                                                                           \\ \cline{3-6} 
                                                                      &                                                                                                                                                                                    & \multicolumn{1}{c}{mIoU} & \multicolumn{1}{c}{mF1} & \multicolumn{1}{c}{mPre} & \multicolumn{1}{c}{mRec} \\ \midrule
3                                                                     & \multirow{4}{*}{\begin{tabular}[c]{@{}c@{}}Experts $=10$,\\ FRL prototypes $=16$,\\ $\lambda_{\text{LB}}=0.01$,\\ $\lambda_{\text{FRL}}=10^{-4}$\end{tabular}} & \bestcell{45.85}                    & \bestcell{56.75}                   & \bestcell{69.91}                    & \bestcell{53.69}                    \\
4                                                                     &                                                                                                                                                                                    & 45.46                    & 55.47                   & 68.82                    & 53.41                    \\
5                                                                     &  & 44.74                     & 55.99                     & 68.30                    & 53.47                                                                                                                                                                                                            \\
6                                                                     &                                                                                                                                                                                    & 44.33                    & 55.65                   & 67.44                    & 52.48                    \\ \bottomrule
\end{tabularx}
\label{tab:abl_topk_value}
\vspace{-0.3cm}
\end{table}

\begin{table}[tp]
\setlength{\tabcolsep}{6.8pt}
\caption{Comparison among cases with various FRL prototype numbers}
\begin{tabularx}{\linewidth}{c|c|cccc}
\toprule
\multirow{2}{*}{\begin{tabular}[c]{@{}c@{}}FRL Prototype\\ Number\end{tabular}} & \multirow{2}{*}{\begin{tabular}[c]{@{}c@{}}Other Fixed\\ Settings\end{tabular}}                                                                              & \multicolumn{4}{c}{Cityscapes} \\ \cline{3-6} 
                                                                                &                                                                                                                                                                   & mIoU   & mF1   & mPre  & mRec  \\ \midrule
4                                                                               & \multirow{4}{*}{\begin{tabular}[c]{@{}c@{}}Experts $=10$,\\ TopK $=5$,\\ $\lambda_{\text{LB}}=0.01$,\\ $\lambda_{\text{FRL}}=10^{-4}$\end{tabular}} & \bestcell{45.59}  & \bestcell{57.35} & \bestcell{70.76} & 53.84 \\
8                                                                               &                                                                                                                                                                   & 45.27  & 56.92 & 69.16 & 53.68 \\
12                                                                              &                                                                                                                                                                   & 44.47  & 55.25 & 67.99 & \bestcell{54.04} \\
16                                                                              &                                                                                                                                                                   & 44.74                     & 55.99                     & 68.30                    & 53.47       \\ \bottomrule
\end{tabularx}
\label{tab:abl_FRL_proto_num}
\vspace{-0.3cm}
\end{table}

\subsubsection{The Effect of the TopK Value on AirMoE Performance}
\Cref{tab:abl_topk_value} examines the routing sparsity by sweeping TopK from $3$ to $6$ while fixing the expert pool at $N=10$. The sparsest configuration of TopK$=3$ dominates every metric (\ie, $45.85$ in mIoU, $56.75$ in mF1, $69.91$ in mPre, $53.69$ in mRec), and increasing $K$ degrades performance steadily until TopK$=6$. The reason is that a larger $K$ forces the router to admit lower-affinity experts beyond the well-matched top ones. These marginal experts contribute noisy logits into the fused output, which is exactly why mPre falls most sharply ($-2.47$ from $K{=}3$ to $K{=}6$) while mRec also erodes ($-1.21$). This confirms that AirMoE benefits from decisive, low-$K$ routing rather than broad ensembling. This is further reinforced from the over-the-air perspective. Each additionally activated expert must transmit its contribution over the shared channel, hence, a larger $K$ raises aggregation noise and power contention under the truncated channel-inversion scaling in \Cref{eq:pc}, compounding the model-side dilution with a communication-side penalty. In summary, AirMoE prefers a small number of routing experts.

\subsubsection{The Impact of the Number of FRL Prototypes on AirMoE Performance}
\Cref{tab:abl_FRL_proto_num} probes the capacity of FRL module by varying its prototype count from $4$ to $16$. The dominant trend here reveals a preference for a small prototype set. The smallest configuration of $4$ prototypes achieves the best mIoU ($45.59$), mF1 ($57.35$), and mPre ($70.76$), and enlarging the prototype bank generally erodes performance. This is consistent with a representation-granularity argument. That is, each prototype acts as a learnable anchor to partition the feature space, and when the number of prototypes exceeds the number of genuinely required modes in the data, the surplus anchors fragment the feature assignment and dilute the discriminative signal that the FRL regularizer is meant to sharpen. The precision column again reflects this most clearly ($-2.77$ from $4$ to $12$ prototypes), since redundant prototypes blur the boundaries between confidently separated classes. Importantly, the degradation is not strictly monotonic. Concretely, mIoU recovers slightly from $44.47$ at $12$ prototypes to $44.74$ at $16$, and mRec is in fact maximized at $12$ ($54.04$). We attribute this mild rebound to a two-regime behavior, wherein an oversized bank first introduces harmful redundancy but a sufficiently large bank ($16$) begins to re-specialize its anchors into finer sub-clusters, partially restoring assignment stability. The recall peak at $12$ likewise indicates that a denser prototype set can marginally broaden region coverage even as it sacrifices precision. In conclusion, a compact prototype set best matches the intrinsic semantic complexity of the task, and we therefore adopt a small prototype number to keep the FRL representation both discriminative and computationally economical.

\begin{table}[tp]
\setlength{\tabcolsep}{7.5pt}
\caption{Comparison among cases with different $\lambda_{\text{LB}}$ values}
\begin{tabularx}{\linewidth}{c|c|llll}
\toprule
\multirow{2}{*}{$\lambda_{\text{LB}}$} & \multirow{2}{*}{\begin{tabular}[c]{@{}c@{}}Other Fixed\\ Settings\end{tabular}}                                                                      & \multicolumn{4}{c}{Cityscapes}                                                                           \\ \cline{3-6} 
                           &                                                                                                                                                           & \multicolumn{1}{c}{mIoU} & \multicolumn{1}{c}{mF1} & \multicolumn{1}{c}{mPre} & \multicolumn{1}{c}{mRec} \\ \midrule
0.00                       & \multirow{5}{*}{\begin{tabular}[c]{@{}c@{}}Experts $=10$,\\ TopK $=5$,\\ FRL prototypes $=16$,\\ $\lambda_{\text{FRL}}=10^{-4}$\end{tabular}} & 44.23                    & 55.48                   & \bestcell{71.13}                    & 50.82                    \\
0.01                       &     & 44.74                     & 55.99                     & 68.30                    & \bestcell{53.47}                          \\
0.02                       &                                                                                                                                                           & 43.97                    & 55.19                   & 68.87                    & 52.06                    \\
0.04                       &                                                                                                                                                           & \bestcell{45.44}                    & \bestcell{57.17}                   & 70.86                    & 53.09                    \\
0.08                       &                                                                                                                                                           & 44.04                    & 55.46                   & 69.30                    & 51.78                    \\ \bottomrule
\end{tabularx}
\label{tab:abl_LB_weight}
\vspace{-0.25cm}
\end{table}

\begin{table}[tp]
\setlength{\tabcolsep}{7.3pt}
\caption{Comparison among cases with different $\lambda_{\text{FRL}}$ values}
\begin{tabularx}{\linewidth}{c|c|cccc}
\toprule
\multirow{2}{*}{$\lambda_{\text{FRL}}$} & \multirow{2}{*}{\begin{tabular}[c]{@{}c@{}}Other Fixed\\ Settings\end{tabular}}                                                                     & \multicolumn{4}{c}{Cityscapes}                                                                           \\ \cline{3-6} 
                            &                                                                                                                                                          & \multicolumn{1}{c}{mIoU} & \multicolumn{1}{c}{mF1} & \multicolumn{1}{c}{mPre} & \multicolumn{1}{c}{mRec} \\ \midrule
0.00                        & \multirow{5}{*}{\begin{tabular}[c]{@{}c@{}}Experts $=10$,\\ TopK $=5$,\\ FRL prototypes $=16$,\\ $\lambda_{\text{LB}}=0.01$\end{tabular}} & 44.23                    & 55.71                   & 68.62                    & 52.18                    \\
1e-4                        &   & 44.74                     & 55.99                     & 68.30                    & 53.47                          \\
2e-4                        &                                                                                                                                                          & 44.16                    & 55.47                   & 69.65                    & 52.19                    \\
4e-4                        &                                                                                                                                                          & \bestcell{45.44}                    & \bestcell{57.17}                   & \bestcell{71.17}                    & \bestcell{53.60}                    \\
8e-4                        &                                                                                                                                                          & 45.03                    & 56.44                   & 69.93                    & 52.29                    \\ \bottomrule
\end{tabularx}
\label{tab:abl_FRL_weight}
\vspace{-0.3cm}
\end{table}

\subsubsection{The Effect of the Load-Balancing Weight $\lambda_{\text{LB}}$ on AirMoE Performance}
\Cref{tab:abl_LB_weight} studies how strongly the load-balancing regularizer $\lambda_{\text{LB}}$ pushes the router toward an even utilization of experts. In contrast to the previous ablations, the response here is distinctly non-monotonic and interior-optimal. The performance instead is maximized at an intermediate value $\lambda_{\text{LB}}=0.04$ ($45.44$ mIoU, $57.17$ mF1), forming an inverted-U profile. The two ends of the sweep are diagnostic. At the end of $\lambda_{\text{LB}}=0$, the router is free to collapse onto a small subset of favored experts, and the metric confirms this expert collapse precisely, where mPre attains its global maximum ($71.13$) while mRec goes to its global minimum ($50.82$). This is the classic overconfident-specialist regime, in which a few dominant experts fire decisively on the classes (high precision) but the starved, rarely-routed experts leave large portions of the label space uncovered (low recall). On the other end, an overly strong constraint ($\lambda_{\text{LB}}=0.08$) begins to force uniform routing even when a token clearly belongs to one expert. This over-regularization homogenizes the experts and drags all metrics back down ($44.04$ mIoU, $51.78$ mRec). Overall, the inverted-U confirms that load balancing is beneficial only in moderation. The chosen regularization strength must be strong enough to prevent expert collapse yet weak enough to preserve the input-conditioned specialization on which AirMoE's accuracy depends.

\subsubsection{The Impact of the FRL Loss Weight $\lambda_{\text{FRL}}$ on AirMoE Performance}
\Cref{tab:abl_FRL_weight} investigates the strength of the FRL regularizer by sweeping $\lambda_{\text{FRL}}$ from $0$ to $8\text{e-}4$. The effect is non-monotonic and interior-optimal, peaking cleanly at $\lambda_{\text{FRL}}=4\text{e-}4$, which uniquely attains the best value on all metrics (\ie, $45.44$ in mIoU, $57.17$ in mF1, $71.17$ in mPre, $53.60$ in mRec). This indicates that the optimal FRL strength improves precision and recall jointly rather than trading one against the other, \ie, it genuinely sharpens the feature space instead of merely shifting the decision threshold. Specifically, at $\lambda_{\text{FRL}}=0$, the FRL objective is inactive, the prototypes are unconstrained, and the features remain comparatively diffuse, yielding the lowest mIoU ($44.23$). As the weight increases toward $4\text{e-}4$, this regularization progressively pulls same-class features toward their prototypes and enforces inter-class separation, lifting every metric to its maximum. However, an overly large $\lambda_{\text{FRL}}=8\text{e-}4$ lets the auxiliary regularizer begin to dominate the primary segmentation loss and over-constrain features toward the prototypes, at the expense of the task-required pixel-level discrimination. Consequently, mRec falls to its lowest value ($52.29$) and mIoU retreats to $45.03$. Taken together, the inverted-U confirms that the FRL regularizer is most effective as a moderate auxiliary signal. It should be strong enough to impose prototype-guided feature clustering, yet subordinate to the main segmentation objective so as not to distort task-critical discriminative cues.

\section{Conclusion} \label{sec:conclusions}
We studied MoE over a cloud--edge wireless network, where a pretrained LM backbone resides at the cloud and specialized experts are distributed across clients. We identified routing and aggregating as two communication bottlenecks with fundamentally different needs. To solve such bottlenecks, we proposed AirMoE, which decouples them both algorithmically and physically. Specifically, MoE-RM performs reliable, low-rate digital routing using compact prototype statistics, while MoE-AM realizes the statistically reweighted expert fusion directly over the multiple-access channel via channel-aware power control, achieving aggregation latency and bandwidth that are invariant to the number of activated experts. Furthermore, theoretical guarantees make AirMoE's dependence on channel quality explicit. Finally, we conducted extensive experiments on semantic segmentation that show the proposed AirMoE obtains consistent gains over MoE and single-model baselines. Current limitations include the mere focus on the idealized analog front-end, and extending AirMoE to incorporate imperfect channel-state information and synchronization errors is a promising direction for future work.

\end{document}